\documentclass{llncs}
\usepackage{amsmath, amssymb}
\usepackage{graphicx}
\usepackage[algo2e, ruled, vlined, linesnumbered]{algorithm2e}
\DontPrintSemicolon
\usepackage{subfigure}
\usepackage{wrapfig}
\usepackage{tabularx}
\usepackage[table]{xcolor}
\usepackage{multirow}
\usepackage{thmtools, thm-restate}
\usepackage{ascii}
\usepackage{paralist}
\usepackage{enumerate}
\usepackage{bibunits}
\usepackage{times}

\newcommand{\algo}[1]{{\normalfont\scshape #1}}
\DeclareMathAlphabet{\mathcal}{OMS}{cmsy}{m}{n}
\SetKwComment{myco}{// }{\DontPrintSemicolon}
\definecolor{RowColor1}{HTML}{FFFFFF}
\definecolor{RowColor2}{HTML}{DDDDDD}
\definecolor{HeaderRowColor}{HTML}{FFFFFF}

\def\comment#1{}

\def\withcomments{
  \newcounter{mycommentcounter}
  \def\comment##1{\refstepcounter{mycommentcounter}%
   \ifhmode%
    \unskip%
    {\dimen1=\baselineskip \divide\dimen1 by 2 %
      \raise\dimen1\llap{\tiny -\themycommentcounter-}}\fi%
    \marginpar{\renewcommand{\baselinestretch}{0.8}%
      \footnotesize [\themycommentcounter]: \raggedright ##1}}
  }
\withcomments

\newenvironment{sketch}{\noindent \textit{Sketch of proof.}\,\,}{\hfill\qed\smallskip}
\newcommand{\rephrase}[3]{\noindent\textbf{#1~#2.~}\emph{#3}}
\newcommand{\T}{\ensuremath{\mathcal T}}

\newcommand{\C}{\ensuremath{\Omega}}

\newcommand{\Ea}{\ensuremath{E_\alpha}}

\newcommand{\caa}{\ensuremath{c_\alpha}}
\newcommand{\Ga}{\ensuremath{G_\alpha}}
\newcommand{\Va}{\ensuremath{V_\alpha}}
\newcommand{\SC}{\ensuremath{\mathcal {SC}}}
\newcommand{\Ts}{\ensuremath{T_\ast}}
\newcommand{\Vs}{\ensuremath{V_\ast}}
\newcommand{\Es}{\ensuremath{E_\ast}}
\newcommand{\css}{\ensuremath{c_\ast}}

\SetKwComment{myco}{// }{\DontPrintSemicolon}
\title{Hierarchies of Predominantly Connected Communities}
\author{
	Michael Hamann,
	Tanja Hartmann\inst{},
	Dorothea Wagner
}
\institute{
  Department of Informatics, Karlsruhe Institute of Technology (KIT)
\\ \url{michael@content-space.de,\{t.hartmann, dorothea.wagner\}@kit.edu}
}
\begin{document}
\setcounter{totalnumber}{8}
\setcounter{topnumber}{8}
\setcounter{bottomnumber}{8}
\renewcommand{\textfraction}{0.001}
\renewcommand{\topfraction}{0.999}
\renewcommand{\bottomfraction}{0.999}
\maketitle
%
\vspace{-0ex}
\begin{abstract}
We consider communities whose vertices are predominantly connected, i.e.,
the vertices in each community are stronger connected to other community members of the same community than to vertices outside the community.
Flake et al. introduced a hierarchical clustering algorithm that finds such predominantly connected communities of different coarseness depending on an input parameter.
We present a simple and efficient method for constructing a clustering hierarchy according to Flake et al.
that supersedes the necessity of choosing feasible parameter values and guarantees the completeness of the resulting hierarchy, i.e., the hierarchy contains all clusterings that can be constructed by the original algorithm for any parameter value.
However, predominantly connected communities are not organized in a single hierarchy.
Thus, we develop a framework that, after precomputing at most $2(n-1)$ maximum flows,
admits 
a linear time construction of a clustering~$\C(S)$ of predominantly connected communities that contains a given community $S$ and
is maximum in the sense that any further clustering of predominantly connected communities that also contains~$S$ is hierarchically nested in $\C(S)$.
We further generalize this construction yielding a clustering with similar properties for~$k$ given communities in $O(kn)$ time.
This admits the analysis of a network's structure with respect to various communities in different hierarchies.
\end{abstract}
\vspace*{-5ex}
%
\section{Introduction}
\vspace*{-1ex}
There exist many 
different approaches to find communities in networks,
many of which are inspired by graph clustering techniques originally developed 
for special applications in fields like physics and biology.
Graph clustering is based on the assumption that the given network is a compound of
dense subgraphs, so called \emph{clusters} or \emph{communities}, 
that are only sparsely connected among each other, 
and aims at finding a \emph{clustering}
that represents these subgraphs.
However, evaluating the quality of a found clustering is often difficult,
since there are no generally applicable criteria for good clusterings
and clustering properties that are well 
interpretable in the network's context are rarely guaranteed. 
In this work we thus focus on predominantly connected communities in 
undirected edge-weighted graphs. Predominant connectivity is 
easy to interpret and guarantees that only vertices whose membership to a community is clearly 
indicated by the networks's structure are assigned to a community.
The latter is in particular desired 
if the analysis of the community structure is meant to support costly or risky decisions.

\vspace{-2ex}
\subsubsection{Contribution and Outline.} We discuss different types 
of predominantly connected communities
(cp.~Table~\ref{tab:tighttypes} for an overview) in Section~\ref{sec:preCom} 
and argue that considering 
source communities (SCs) in networks is reasonable. 
We further give a characterization of SCs and 
introduce basic nesting properties.
In Section~\ref{sec:cutClus}, we review the 
cut clustering algorithm by Flake et al.~\cite{ftt-gcmct-04}, 
which takes an input parameter $\alpha$
and decomposes a given network into SCs, 
each of which providing an intra-cluster density of at least $\alpha$ and 
an inter-cluster sparsity of at most $\alpha$.
At the same time, $\alpha$ controls the 
coarseness of the resulting clustering such that for
varying values the algorithm returns a clustering hierarchy.
Flake at al.\ refer to Gallo et al.~\cite{ggt-fpmfaa-89} for the
question how to choose~$\alpha$ such that all possible hierarchy levels are found.
However, they give no further description how to extend the approach of Gallo et al.,
which finds all breakpoints of $\alpha$ for a single parametric flow, 
to a fast construction of a complete hierarchy. 
They just propose a binary-search approach to find good values for~$\alpha$.
We introduce a parametric-search approach that guarantees the completeness of the resulting hierarchy 
and exceeds the running time of a
binary search-based approach, whose running time strongly depends on the
discretization of the parameter range.
\renewcommand{\arraystretch}{1.2}
\setlength{\tabcolsep}{1mm}
\begin{table}[t]
	\centering
	\caption{Overview of different types of predominantly connected communities. The columns to the right describe the relations between the types in terms of inclusion.}
	\vspace{-1ex}
	\begin{tabular}{|c|c|c||c|c|c|}
 		\hline
		\multicolumn{3}{|c||}{\textrm{A subgraph } $S \subseteq V$ \textrm{ is a} }& WC & ES & SC\\ 
		\hline
		\textrm{WC} & $\forall u\in S$ & $c(\{u\},S\setminus \{u\}) \gneqq c(\{u\}, V\setminus S)$  & x & x & \\
		\cline{1-6}
		\textrm{ES}  &   $\forall U\subset S$ & $c(U,S \setminus U) \gneqq c(U, V \setminus S)$ &  & x &  \\
		\cline{1-6}
		\textrm{SC} & $\exists s\in S: \forall U\subset S$, $s\notin U$ & $c(U,S\setminus U) \gneqq c(U, V\setminus S)$  &  & x &  x\\
		\hline
	\end{tabular}
	\label{tab:tighttypes}
\end{table}

\vspace{0ex}
Experimental evaluations further showed 
that the cut clustering algorithm finds meaningful clusters in real-world instances~\cite{ftt-gcmct-04},
but yet,
it often happens that even in a complete hierarchy non-singleton clusters are only found for a subgraph 
of the initial network, while the remaining vertices stay unclustered
even on the coarsest non-trivial hierarchy level~\cite{ftt-gcmct-04,hhw-chcca-13}.
Motivated by this observation, in Section~\ref{sec:framework}, we develop a framework 
that is based on a set $M(G)$ of $n \leq |M(G)| \leq 2(n-1)$ maximal SCs in the graph~$G$, 
i.e., each further SC is nested in a SC in $M(G)$,
and is represented by a special cut tree, which can be constructed by at 
most $2(n-1)$ max-flow computations.
After computing $M(G)$ in a preprocessing step,
the framework efficiently answers the following queries: 
(i) Given an arbitrary SC $S$, 
what does a clustering $\C(S)$ look like that consists of~$S$ and further SCs such that any SC not intersecting with $S$ is nested in a cluster of $\C(S)$? In particular, $\C(S)$ is maximum in the sense that any clustering of SCs that contains~$S$ is hierarchically nested in $\C(S)$.
We show that $\C(S)$ can be determined in linear time.
(ii) Given $k$ disjoint SCs, which is the maximal clustering $\C(S_1,\dots, S_k)$ that contains the given SCs, is nested in each $\C(S_i)$, $i = 1,\dots, k$, and
guarantees that any clustering of SCs that also contains the given ones is nested in $\C(S_1,\dots, S_k)$?
Computing $\C(S_1,\dots, S_k)$ takes $O(kn)$ time.
These queries allow to further examine the
community structure of a given network, beyond the complete clustering hierarchy according to Flake et al.
We exemplarily apply both queries to a small real world network, 
thereby finding a new clustering beyond the hierarchy 
that contains all non-singleton clusters of the best clustering in the 
hierarchy but far less singletons.

\vspace*{-1.5ex}
\subsubsection{Preliminaries.}
Throughout this work we consider an undirected, weighted graph $G = (V,E,c)$ with vertices $V$, edges $E$ and a positive edge cost function~$c$, writing $c(u,v)$ as a shorthand for $c(\{u,v\})$ with 
$\{u,v\} \in E$. Whenever we consider the degree $\deg(v)$ of $v\in V$, we implicitly mean the sum of all edge costs incident to~$v$.
A \emph{cut} in $G$ is a partition of $V$ into two \emph{cut sides} $S$ and $V\setminus S$. The cost $c(S,V\setminus S)$ of a cut is the sum of the costs of all edges \emph{crossing} the cut, i.e., edges $\{u,v\}$ with $u\in S$, $v \in V\setminus S$. For two disjoint sets $A,B\subseteq V$ we define the cost $c(A,B)$ analogously. 
Two cuts are \emph{non-crossing} if their cut sides are pairwise nested or disjoint. 
Two sets $S,T \subset V$ are \emph{separated} by a cut if they lie on different cut sides.
A minimum $S$-$T$-cut is a cut that separates $S$ and $T$ and is the cheapest cut among all cuts separating these sets. We call a cut a \emph{minimum separating cut} if there exists an arbitrary pair $\{S,T\}$ for which it is a minimum $S$-$T$-cut.  We identify singleton sets with the contained vertex without further notice.
We further denote the \emph{connectivity} of $\{S,T\}\subseteq 2^V$ by $\lambda(S,T)$, describing the cost of a minimum $S$-$T$-cut.
A \emph{clustering} $\C$ of $G$ is a partition of $V$ into subsets $C^1, \dots , C^k$, which define vertex-induced subgraphs, called \emph{clusters}.
A cluster is \emph{trivial} if it corresponds to a connected component. A vertex that forms a singleton cluster although it is no singleton in $G$, is \emph{unclustered}. 
A clustering is \emph{trivial} if it consists of trivial clusters or if $k=n$.
A \emph{hierarchy of clusterings} is a sequence $\C_1 \leq \dots \leq \C_r$ such that $\C_i \leq \C_j$ implies that each cluster in $\C_i$ is a subset of a cluster in $\C_j$. We say $\C_i \leq \C_{j}$ are \emph{hierarchically nested}.
A clustering $\C$ is \emph{maximal} with respect to a property $\mathcal P$ if there is no other clustering $\C'$ with property $\mathcal P$ and $\C \leq \C'$.

\vspace*{-1ex}
\section{Predominantly Connected Communities}\label{sec:preCom}
\vspace*{-1ex}
In the context of large web-based graphs, Flake et al.~\cite{flgc-soiwc-02} 
introduce \emph{web communities} (WCs)
in terms of predominant connectivity of single vertices:
A set
 $S\subseteq V$ is a web community
if $c(\{u\},S\setminus \{u\}) \gneqq c(\{u\}, V\setminus S)$ for all $u \in S$.
Web communities are
not necessarily connected (cp.~Fig.~\ref{fig:webCom})
and decomposing a graph into $k$ web communities is NP-complete~\cite{ftt-gcmct-04}.
Extending
\begin{wrapfigure}[6]{r}{0.25\textwidth}
\vspace{-4.5ex}
\centering
\includegraphics[width = 2.5cm]{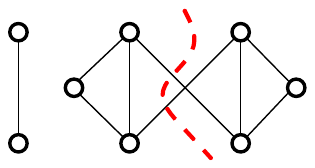}
\vspace{-1.5ex}
\caption{\small Unconnected web community (left) }
\label{fig:webCom}
\end{wrapfigure}
the predominant connectivity from vertices to arbitrary subsets
yields \emph{extreme sets} (ESs), which satisfy a stricter property 
that guarantees connectivity and
gives a good intuition why the vertices in ESs belong together:
A set $S\subseteq V$ is an extreme set
if $c(U,S\setminus U) \gneqq c(U, V\setminus S)$ for all $U \subsetneq S$.
The extreme sets in a graph can be computed in $O(nm + n^2 \log n)$ time
with the help of maximum adjacency orderings~\cite{n-gancp-04}. 
They form a subset of the \emph{maximal components} of a graph,
which subsume vertices that are not separated by cuts 
cheaper than a certain lower bound. Maximal components
are either nested or disjoint and can be deduced from a cut tree, 
whose construction needs $n-1$ maximum flow computations~\cite{gh-mtnf-61}.
They are used in the context of image segmentation by Wu and Leahy~\cite{wl-aogta-93}.

In, for example, social networks, we are also interested in communities 
that surround a designated vertex, for instance a central person.
Complying with this view, 
\emph{source communities} (SCs) describe vertex sets where each subset
that does not contain a designated vertex is predominantly connected to the remainder of the group:
A set $S\subseteq V$ is a SC with \emph{source} $s\in S$
if $c(U,S\setminus U) \gneqq c(U, V\setminus S)$ for all $U \subsetneq S\setminus \{s\}$.
The members of a SC can be interpreted as \emph{followers} of the source 
in that sense that each subgroup feels more attracted by the source (and  
other group members) than by the vertices outside the group.
The predominant connectivity of SCs implements a close relation 
to minimum separating cuts. 
In fact, SCs are characterized as follows.
\newcommand{\lemCharSC}{
A set $S \subset V$ is a SC  of $s \in S$ iff 
there is $T \subseteq V\setminus S$ such that $(S,V\setminus S)$ is 
the minimum $s$-$T$-cut in $G$ that minimizes the number of vertices on the side containing~$s$.
}
\begin{lemma}
  \label{lem:charSC}
  \lemCharSC
\end{lemma}
\begin{proof}
($\Rightarrow$): If $S$ is a SC of $s$, $(S,V\setminus S)$ is a 
minimum $s$-$T$-cut for $T = V\setminus S$. Otherwise, a cheaper
$s$-$T$-cut would split $S$ into $U$ and $S\setminus U \ni s$ with
$c(U, S\setminus U) < c(U, V\setminus S)$, which is a contradiction.
The cut $(S,V\setminus S)$ further minimizes the number of vertices 
on the side containing $s$, since otherwise a "smaller" cut would induce a
set $U \subsetneq S$, $s\notin U$, with
$c(U, S\setminus U) = c(U, V\setminus S)$.

($\Leftarrow$): If $(S,V\setminus S)$ is a minimum $s$-$T$-cut with $T\subseteq V\setminus S$,
then it is $c(U, S\setminus U) \leq c(U, V\setminus S)$ for all $U\subset S\setminus \{s\}$. 
Otherwise, $(S\setminus U, V\setminus (S\setminus U))$, which also separates~$s$ and~$T$, 
would be a cheaper $s$-$T$-cut.
If $(S,V\setminus S)$ further minimizes the number of vertices 
on the side containing $s$, it is $c(U, S\setminus U) < c(U, V\setminus S)$ for all $U\subset S\setminus \{s\}$. 
Otherwise, $(S\setminus U, V\setminus (S\setminus U))$, would be a minimum $s$-$T$-cut with a smaller side containing $s$.
\qed
\end{proof}

\vspace*{-0.2ex}
\begin{wrapfigure}[10]{r}{0.21\textwidth}
\vspace{-6ex}
\centering
\includegraphics[width = 2.3cm]{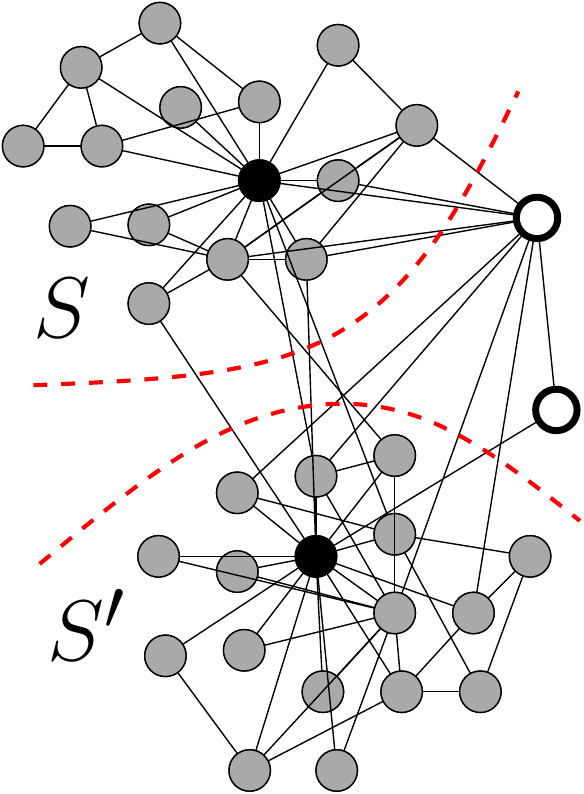}
\label{fig:zach}
\vspace{-1.5ex}
\caption{\small Indecisive vertices (white). }
\end{wrapfigure}
Based on this characterization, we introduce some further notations and 
two basic lemmas on nesting properties of SCs, which we will mainly use in Section~\ref{sec:framework}.
Note that a minimum $s$-$T$-cut in~$G$ must not be unique, 
however, the minimum $s$-$T$-cut that minimizes the number of vertices on the side containing~$s$ is unique.
We call such a cut, which induces a SC~$S$, a \emph{community cut}, 
$S$ the SC of $s$ \emph{with respect to~$T$} and $T$ the \emph{opponent} of $s$.
Hence, $\SC\colon V\times 2^V \rightarrow 2^V$, $\SC(s,T) \mapsto$ \{the SC of~$s$ with respect to~$T$\}
is well defined providing $\SC(s,T)$ as future notation.
The corresponding maximum flow between $s$ and~$T$ also induces an \emph{opposite} SC $S' := \SC(T,s)$,
if we consider~$T$ as a compound node.
If the community cut is the only minimum $s$-$T$-cut, it is $S' = V\setminus S$. 
Otherwise, $X:= V\setminus (S\cup S') \not= \emptyset$ and the vertices in $X$  are 
neither predominantly connected within $S\cup X$ nor within $S'\cup X$, i.e.,
$c(U,S\cup X) \leq c(U, V\setminus (S\cup X))$ for all $U \subseteq X$ (analogously for $S'$).
In, for example, a social network this can be interpreted as follows.
Whenever $s$ and the group $T$ become rivals, 
the network decomposes into followers of $s$ (in $S$), followers of $T$ (in $S'$) 
and possibly some \emph{indecisive} individuals in $V\setminus (S\cup S')$. 
Figure~\ref{fig:zach} exemplarily shows two indecisive vertices in the (unweighted\footnote{Zachary considers the weighted network and therein the minimum cut that separates the two central vertices of highest degree (black). In the weighted network this cut is unique.}) karate club network 
gathered by Zachary~\cite{z-ifmcf-77}. 
Note that a SC can have several sources, and a vertex can have different SCs w.r.t.\ different
opponents.
The SCs of a vertex are partially nested as stated in Lemma~\ref{lem:comNest}, 
which is a special case of (2i) of  Lemma~\ref{lem:intersecBehavior} summarizing the intersection behavior of arbitrary SCs.
See Figure~\ref{fig:intersecBehavior} for illustration and an example of neither nested nor disjoint~SCs.
\newcommand{\lemComNest}{
Let $S$ denote a SC of $s$ and $T\cap S = \emptyset$. 
Then $S\subseteq \SC(s,T)$.
}
\begin{lemma}
  \label{lem:comNest}
  \lemComNest
\end{lemma}
\begin{proof}
Since Case (2i) of Lemma~\ref{lem:intersecBehavior} admits $s_1 = s_2$, 
this lemma directly follows with $s \equiv s_1$, $S \equiv S_1$, $T\equiv T_2$ and $S' \equiv S_2$.
\qed
\end{proof}
As a consequence, each SC $S\not= V$ is nested in a SC $S'$ that is a SC w.r.t.\ a single vertex~$t$,
while any SC $\bar S$ with $S'\subsetneqq \bar S$ contains~$t$.
In this sense, SCs w.r.t.\ single vertices are \emph{maximal}.
We denote the set of maximal SCs in $G$ by $M(G)$.
\newcommand{\intersecBehavior}{
Consider $S_1 := \SC(s_1,T_1)$ and $S_2 := \SC(s_2,T_2)$.
\begin{compactenum}[(1)]
\item If $\{s_1,s_2\}\cap (S_1\cap S_2) = \emptyset$, then $S_1 \cap S_2 = \emptyset$.
\item If $T_2\cap S_1 = \emptyset$ and $s_1 \in S_2$, then $S_1\subseteq S_2$ (i).
If further $T_1 \cap S_2 = \emptyset$ and $s_2 \in S_1$, then $S_1 = S_2$ (ii).
\item Otherwise, $S_1$ and $S_2$ are neither nested nor disjoint.
\end{compactenum}
}
\begin{lemma}
  \label{lem:intersecBehavior}
  \intersecBehavior
\end{lemma}
\begin{proof}
Proof of (1):\\
First recall that $S_1\cap T_1 = S_2\cap T_2 = \emptyset$.
Now suppose $S_1\cap S_2 \not= \emptyset$ (cp.Figure~\ref{fig:Case_i}).
Since $U:= (S_1 \cap S_2) \subseteq S_1$ with $s_1\notin U$ it is
$\left. \deg(U) \right|_{S_1} \gneqq \left. \deg(U) \right|_{(V\setminus S_1)\cup U}$. This is equivalent to
$c(U,V\setminus S_1) < c(U, S_1\setminus U)$ and it follows, since $(S_2\setminus U) \subseteq (V\setminus S_1)$,
\begin{eqnarray}
c(U,S_2\setminus U) &\leq& c(U,V\setminus S_1) < c(U, S_1\setminus U). \label{eq:help}
\end{eqnarray}
We apply inequality~(\ref{eq:help}) in order to show that the cut $(S_2\setminus U, V\setminus (S_2\setminus U))$,
which also separates $s_2$ and $T_2$, is cheaper than
the community cut inducing $S_2$, which leads to a contradiction.

We represent the costs of the two cuts as follows:
\begin{eqnarray}
c(S_2, V\setminus S_2) = & & \nonumber \\
c(S_2\setminus U, S_1\setminus U) & + & c(U,S_1\setminus U) + c(S_2, V\setminus (S_1\cup S_2)) \label{cut1}\\
c(S_2\setminus U, V\setminus (S_2\setminus U)) =& &\nonumber \\
c(S_2\setminus U, S_1\setminus U) &+& c(S_2\setminus U,U) + c(S_2\setminus U, V\setminus (S_1\cup S_2)) \label{cut2}\
\end{eqnarray}
Since $(S_2\setminus U)\subseteq S_2$ it is $c(S_2\setminus U, V\setminus (S_1\cup S_2)) \leq c(S_2, V\setminus (S_1\cup S_2))$
and with~(\ref{eq:help}) we see that~(\ref{cut1}) $<$~(\ref{cut2}). 

Proof of (2i):\\
Suppose $S_1\setminus S_2 \not=\emptyset$ (cp.~Fig~\ref{fig:TtwoSone_g}).
Then it is
$c(U,S_2) \leq c(U, V\setminus (S_1\cup S_2))$, 
since otherwise $(S_1\cup S_2, V\setminus(S_1\cup S_2))$ would be a cheaper $s_2$-$T_2$-cut
than the community cut inducing $S_2$.
Since $(S_1\setminus U) \cup (S_2\setminus S_1) = S_2$, it is
\begin{eqnarray}
c(U,S_1\setminus U) + c(U,S_2\setminus S_1) &\leq& c(U,V\setminus (S_1\cap S_2)). \label{eq:help2}
\end{eqnarray}
We apply inequality~(\ref{eq:help2}) in order to show that the cut $(S_1\setminus U,V\setminus (S_1\setminus U))$,
which also separates $s_1$ and $T_1$,
is at most as expansive
as the community cut inducing $S_1$, which leads to a contradiction, since $|S_1\setminus U| < |S_1|$.

We represent the costs of the two cuts as follows:
\begin{eqnarray}
c(S_1, V\setminus S_1) = & & \nonumber \\
c(S_1\setminus U, S_2 \setminus S_1) & + & c(U,S_2\setminus S_1) + c(S_1\setminus U, V\setminus (S_1\cup S_2))\nonumber\\
&+ & c(U,V\setminus(S_1\cup S_2)) \label{cut3}\\
c(S_1\setminus U, V\setminus (S_1\setminus U)) =& &\nonumber \\
c(S_1\setminus U, S_2\setminus S_1) &+& c(S_1\setminus U,U) + c(S_1\setminus U, V\setminus (S_1\cup S_2)) \label{cut4}\
\end{eqnarray}
If we add $c(U,S_2\setminus S_1)$ to~\ref{cut4} and apply~(\ref{eq:help2}) we get a result that is at most as expansive than~$\ref{cut3}$.
Hence,~(\ref{cut4}) $\leq$ (\ref{cut3}).
But $S_1\setminus U$ is smaller than $S_1$ contradicting the fact that $S_1$ is a source community.

Proof of (2ii):\\
Since the general case applies, it is $S_1\subseteq S_2$ (cp.~Fig~\ref{fig:TtwoSone_s}).
Furthermore, with $S_1$ also separating $s_2$ and $T_2$ and $S_2$ also separating $s_1$ and $T_1$ we get $\lambda(s_1, T_1) = \lambda(s_2, T_2)$,
and thus, $(S_1,V\setminus S_1)$ is also a minimum $s_2$-$T_2$-cut with $|S_1|\leq|S_2|$.
Hence, it must be $S_1 = S_2$, otherwise $S_2$ would not be the source community of $s_2$ with respect to $T_2$.

Proof of (3):\\
In the remaining cases it is $S_1\cap S_2 = \emptyset$. Hence, $S_1$ and $S_2$ are not disjoint.
Furthermore, it is either $T_1\cap S_2\not= \emptyset$ and $T_2 \cap S_1 \not= \emptyset$
or $T_1\cap S_2\not= \emptyset$ and $s_1\in S_1\setminus S_2$.
Thus, $S_1$ and $S_2$ are not nested.
Figure~\ref{fig:intersecEx} shows an example where $S_1$ and $S_2$ exist and are neither nested nor disjoint.
\qed
\end{proof}
\begin{figure}[bt]
\centering
\subfigure[Case (1)]{
		\label{fig:Case_i}
		\includegraphics[width = 2.5cm, page=1]{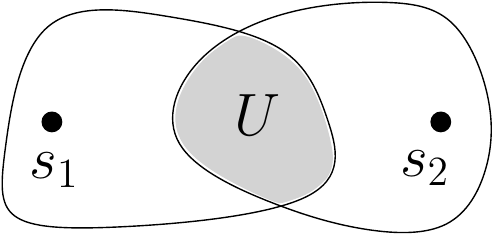}
	}
\hspace{1ex}
	\subfigure[Case (2i)]{
		\label{fig:TtwoSone_g}
		\includegraphics[width = 2.5cm, page=1]{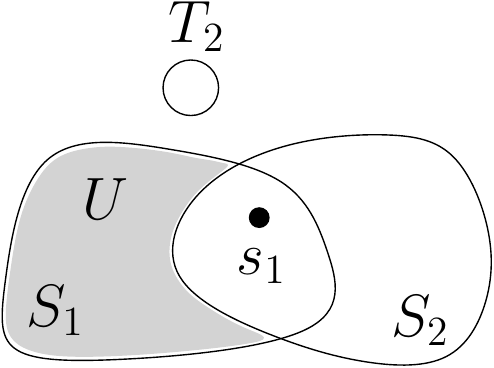}
	}
\hspace{1ex}
	\subfigure[Case (2ii)]{
		\label{fig:TtwoSone_s}
		\includegraphics[width = 2.5cm, page=2]{TtwoSone.pdf}
	}
\hspace{1ex}
	\subfigure[Case (3)]{
		\label{fig:intersecEx}
		\includegraphics[width = 2.3cm, page=1]{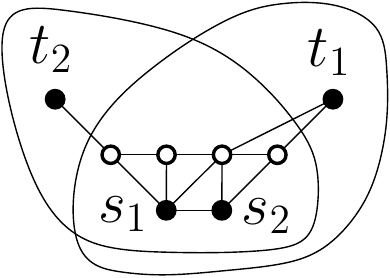}
	}
	\vspace*{-2ex}
	\caption{Situation in Lemma~\ref{lem:intersecBehavior}.}
	\label{fig:intersecBehavior}
\vspace{-0ex}
\end{figure}
\vspace*{-1.5ex}
\section{Complete Hierarchical Cut Clustering}\label{sec:cutClus}
\vspace*{-1ex}
The clustering algorithm of Flake et al.~\cite{ftt-gcmct-04} exploits the properties of minimum separating cuts together with a parameter $\alpha$ 
in order to get clusterings where the clusters are SCs
with the following additional property:
For each cluster $C \in \C$ and each $U\subsetneq C$ it holds
\vspace{-2ex}
\begin{eqnarray}
\frac{c(C,V\setminus C)}{|V\setminus C|}\leq & \alpha& \leq \frac{c(U, C\setminus U)}{\min\{|U|,|C\setminus U|\}} \label{eq:quality}
\end{eqnarray}
According to the left side of this inequality separating a cluster $C$ from the rest of the graph costs at most 
$\alpha |V\setminus C|$ which guarantees a certain inter-cluster sparsity. 
The right side further guarantees a good intra-cluster density in terms of expansion, a measure introduced by~\cite{kvv-cgds-00},
saying that splitting a cluster~$C$ into~$U$ and $C\setminus U$ costs at least $\alpha \min\{|U|,|C\setminus U|\}$.
Hence, the vertex sets representing valid candidates for clusters must be very tight---in addition 
to the predominant connectivity they must also provide an expansion that exceeds a given bound.

\begin{wrapfigure}[13]{r}{.49\textwidth}
\vspace*{-7ex}
\begin{algorithm2e}[H]
		\caption{\algo{CutC}}
		\label{alg:CutC}
		\SetKwComment{tco}{\%}{}
		\KwIn{Graph $\Ga =(\Va,\Ea,\caa)$}
		$\C \gets \emptyset$\;
		\While{$\exists\; u \in V_\alpha \setminus \{t\}$}{\nllabel{ln:simpleCC:while}
		  $C^u \gets$ $\SC(u,t)$ in $\Ga$\;  
		  $r(C^u) \gets u$  \nllabel{ln:simpleCC:flowCC}\;
		  \ForAll{$C^i \in \C$}{
		    \lIf{$r(C^i) \in C^u$} {$\C \gets \C\setminus \{C^i\}$ }
		 }		
		$\C \gets \C \cup \{C^u\}\,$; $V_\alpha \gets V_\alpha \setminus C^u$  \nllabel{ln:simpleCC:b}
		}
		\Return $\mathcal C$
\end{algorithm2e}
\end{wrapfigure}
Flake et al.\ develop their parametric cut clustering algorithm step by step starting from an idea involving cut trees~\cite{gh-mtnf-61}. The final approach, however, just uses community-cuts in a modified graph in order to identify clusters that satisfy condition~(\ref{eq:quality}). 
We refer to this approach by CutC. Here we give a more direct description of this method.
Given a graph $G = (V,E,c)$ and a parameter $\alpha > 0$, as a preprocessing step, augment $G$ by inserting an artificial vertex $t$ and connecting $t$ to each vertex in $G$ by an edge of cost~$\alpha$. Denote the resulting graph by $G_\alpha = (V_\alpha, E_\alpha, c_\alpha)$. Then apply CutC (Alg.~\ref{alg:CutC}) by iterating~$V$ and computing $\SC(u,t)$ for each vertex $u$ not yet contained in a previously computed community. The source $u$ becomes the representative of the newly computed SC (line~\ref{ln:simpleCC:flowCC}).
Since SCs with respect to a common vertex $t$ are either disjoint or nested (Lemma~\ref{lem:intersecBehavior}(1),(2i)), we finally get a set $\C$ of SCs
in $G_\alpha$, which together decompose~$V$. 
Since the vertices in $G_\alpha$ are additionally connected to $t$, 
each SC in $G_\alpha$ with respect to~$t$ is also a SC in $G$.
However it is not necessarily a maximal SC in $M(G)$. 

Applying CutC iteratively with decreasing $\alpha$ 
yields a hierarchy of at most~$n$ different clusterings (cp.~Figure~\ref{fig:hierarcy}). 
This is due to a special nesting property for different parameter values. 
Let~$C_1$ denote the SC of~$u$ in~$G_{\alpha_1}$ and~$C_2$ the SC of~$u$ in~$G_{\alpha_2}$.
Then it is
 $C_1 \subseteq C_2$ if $\alpha_1 \geq \alpha_2$.
The hierarchy is bounded by two trivial clusterings, which we already know in advance. The clustering at the top consists of the connected components of $G$ and 
is returned by CutC for $\alpha_{\max} = 0$, the clustering at the bottom
consists of singletons and comes up if we choose 
$\alpha_0$ equal to the maximum edge cost in~$G$.

\vspace*{-1ex}
\subsubsection{Simple Parametric Search Approach.}
The crucial point with the construction of such a hierarchy, 
however, is the choice of~$\alpha$. If we choose the next 
value too close to a previous one, we get a clustering we 
already know, which implies unnecessary effort. If we 
choose the next value too far from any previous, we 
possibly miss a clustering. Flake et al.\ propose a binary 
search for the choice of $\alpha$. However, this necessitates 
a discretization of the parameter range---an issue where 
again limiting the risk of missing interesting values by small 
steps is opposed to improving the running time by wide steps.
In practise the choice of a good coarseness of the discretization 
requires previous knowledge on the graph structure,
which we usually do not have.  Thus, we introduce a simple 
parametric search approach for constructing a 
complete\footnote{The completeness refers to all clusterings 
that can be obtained by CutC for a value~$\alpha$.} hierarchy 
that does not require any previous knowledge.

\begin{wrapfigure}[10]{r}{.48\textwidth}
\vspace{-6ex}
\centering
\includegraphics[width = 5.5cm]{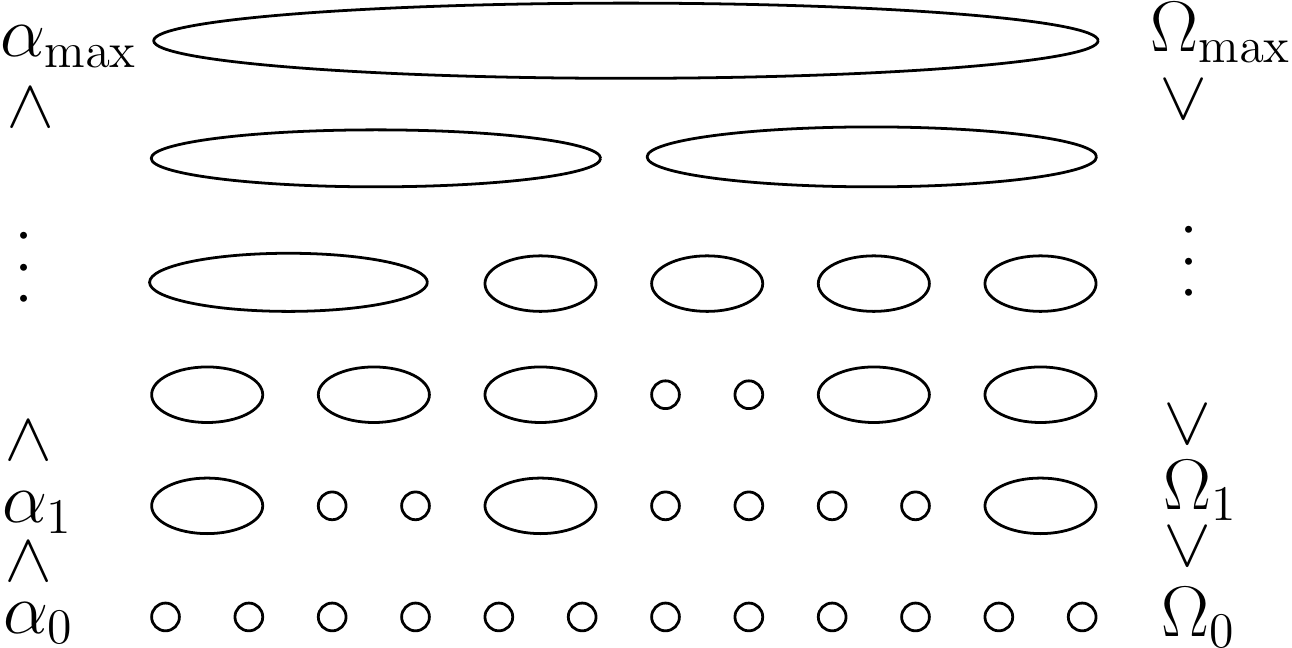}
\caption{Clustering hierarchy by CutC. Note, $\alpha_{\max} < \alpha_0$ whereas $\C_{\max} > \C_0$. }
\label{fig:hierarcy}
\end{wrapfigure}

For two consecutive hierarchy levels $\C_i < \C_{i+1}$ we 
call $\alpha'$ the \emph{breakpoint} if CutC returns $\C_i$ 
for $\alpha'$ and $\C_{i+1}$ for $\alpha' - \varepsilon$ 
with $\varepsilon \rightarrow 0$.
The simple idea of our approach is to compute
good candidates for breakpoints during a recursive search
with the help of cut-cost functions of the clusters,
such that each candidate that is no breakpoint yields a
new clustering instead. In this way, we apply CutC at most 
twice per level in the final hierarchy.
Beginning with the trivial 
clusterings $\C_0 < \C_{\max}$ ($\alpha_0 > \alpha_{\max} $),
the following theorem directly implies an efficient 
algorithm.
\newcommand{\paraApp}{
Let $\C_{i} < \C_{j}$ denote two different clusterings with 
parameter values $\alpha_i > \alpha_j$. 
In time $O(|\C_i|)$ a parameter value $\alpha_m$ with 
1) $\alpha_j < \alpha_m \leq \alpha_i$ can be computed such that 
2)~$\C_{i} \leq \C_{m} < \C_{j} $, and 3)~$\C_{m} = \C_{i}$ implies 
that $\alpha_m$ is the breakpoint between $\C_{i}$ and $\C_{j}$.
}
\begin{theorem}
\label{theo:main}
\paraApp
\end{theorem}

\newpage
\begin{sketch}
We use \emph{cut-cost functions} that represent, 
depending on $\alpha$, the cost $\omega_S(\alpha)$ of 
a cut $(S, V_\alpha \setminus S)$ in $G_\alpha$ based on 
the cost of the cut $(S, V\setminus S)$ in~$G$ and the size of $S$.
\vspace{-3ex}
\begin{eqnarray}
\omega_S : \mathbb R^+_0 & \longrightarrow & [c(S,V \setminus S), \infty) \subset \mathbb R^+_0 \nonumber \\
\omega_S (\alpha) & := & c(S,V \setminus S) + |S|\;\alpha \nonumber
\end{eqnarray}
The main idea is the following. Let $\C_i < \C_j$ denote 
two hierarchically nested clusterings. 
We call a cluster $C' \in \C_i$ that is nested in $C\in \C_j$ a 
child of $C$ and $C$ the parent of $C'$.
 If there exists another level $\C'$ between
$\C_i$ and $\C_j$, at least two clusters in $\C_i$ must be 
merged yielding a larger cluster in $\C'$. The maximal 
parameter value where this happens is a value $\alpha^*$ 
where a child~$C'$ in~$\C_i$ becomes more expensive 
than its parent $C$ in~$\C_j$, and thus, is dominated by $C$ 
in the sense that it will not become a cluster in any 
hierarchy level above~$\alpha^*$ (i.e., where $\alpha < \alpha^*$).
For two nested clusters $C'\subseteq C$ this point is 
marked by the intersection point of the cut-cost functions $\omega_{C'}$ 
and~$\omega_{C}$ (Figure~\ref{fig:cutfunc}).
Thus, this intersection point is a good 
candidate for a breakpoint between $\C_i$ and $\C'$.
We choose
$\alpha_m := \min_{C\in\C_j }\lambda_{C}$ with 
$\lambda_{C} := \max_{C'\in \C_i: C'\subset C}\{\alpha\mid \omega_{C}(\alpha) = \omega_{C'}(\alpha) \} $ 
and prove that Claim 1) to 3) as stated in Theorem~\ref{theo:main} 
hold with this choice of~$\alpha_m$. The proofs are rather technical, 
thus we postpone them to Appendix~\ref{app:CutClus}.

For the running time, observe that $\alpha_m$ is well-defined as each parent function intersects with at least one child function.
In practice we construct $\alpha_m$
by iterating the list of representatives stored for
$\C_i$. These representatives are assigned to a cluster in $\C_j$, thus, matching children to their parents can be done in time $O(|\C_i|)$. The computation of the intersection points takes only constant time, given that the sizes and costs of the clusters are stored with the representatives by CutC. In total, the time for computing $\alpha_m$ is thus in $O(|\C_i|)$.
\end{sketch}
\vspace*{-2ex}
\subsubsection{Running time.}
\begin{wrapfigure}[8]{r}{.25\textwidth}
\vspace{-6ex}
\centering
\includegraphics[width = 2.3cm]{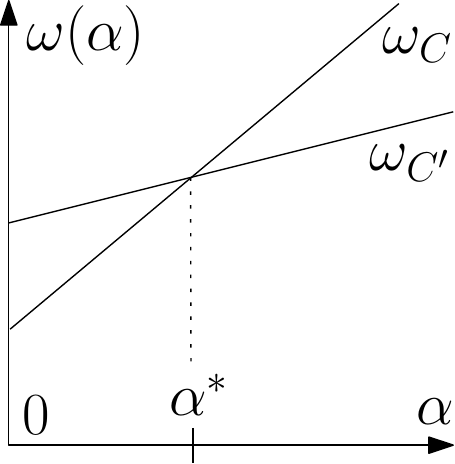}
\vspace{-1.5ex}
\caption{\small Intersecting cut-cost functions. }
\label{fig:cutfunc}
\end{wrapfigure}
The parametric search approach calls CutC twice per level in the final hierarchy, once when computing a level the first time and again right before detecting that the level already exists and a breakpoint is reached.
The trivial levels $\C_{\max}$ and $\C_0$ are calculated in advance without using CutC. Nevertheless, $\C_0$ is recalculated once when the breakpoint to the lowest non-trivial level is found. This yields $2(h-2)+1$ applications of CutC, with $h$ the number of levels.
We denote the running time of CutC by~$T(n)$ without further analysis.
For a more detailed discussion on the running time of CutC see~\cite{ftt-gcmct-04}. 
Since common min-cut algorithms run in $O(n^2 \sqrt{m})$ time, a single min-cut computation already dominates
the costs for determining $\alpha_m$ and further linear overhead. The running time of our simple parametric approach thus is in $O(2h\,T(n))$, where $h\leq n-1$. This obviously improves the running time of a binary search, which is in $O(h\, \log(d)\,T(n))$, with $d$ the number of discretization steps---in particular since we may assume $d\gg n$ in order to minimize the risk of missing levels.
We also tested the practicability of our simple approach by a brief experiment.
The results confirm the improved theoretical running time. We provide them in 
Appendix~\ref{app:CutClus} as bonus. 

\vspace*{-1ex}
\section{Framework for Analyzing SC Structures}\label{sec:framework}
\vspace*{-1ex}
In general, clusterings in which all clusters are SCs 
are only partially hierarchically ordered.
Hence, hierarchical algorithms like 
the cut clustering algorithm of Flake et al.~\cite{ftt-gcmct-04} 
provide only a limited view on the whole 
SC structure of a network.
In this section we develop a framework for efficiently 
analyzing different hierarchies in the SC
structure after precomputing at most
$2(n-1)$ maximum flows. The basis of our framework is the set
$M(G)$ of maximal SCs in $G$.
This can be 
represented by a
cut tree of special community cuts,
together with some additionally stored SCs,
as we will show in the following.

A (general) \emph{cut tree} is a weighted tree 
$\T(G) = (V, E_\T,c_\T)$ on the vertices of an undirected,
weighted graph $G = (V,E,c)$ (with edges not necessarily in $G$) 
such that each $\{s,t\} \in E_\T$ induces a minimum $s$-$t$-cut in $G$ 
(by decomposing $\T(G)$ into two connected components) and 
such that $c_\T(\{s,t\})$ is equal to the cost of the induced cut.
The cut tree algorithm, which was first introduced by Gomory and 
Hu~\cite{gh-mtnf-61} in their pioneering work on cut trees and 
later simplified by Gusfield~\cite{g-vsmap-90}, 
applies $n-1$ cut computations.
For a detailed description of this algorithm see~\cite{gh-mtnf-61,g-vsmap-90} or Appendix~\ref{app:cutTree}.

The main idea of the cut tree algorithm is to iteratively choose 
vertices~$s$ and~$t$ that are not yet separated by a previous 
cut, and separating them by a minimum $s$-$t$-cut, which 
is represented by a new tree edge $\{s,t\}$. 
Depending on the shape of the found cut it might be 
necessary to reconnect previous edges in the intermediate tree. 
Gomory and Hu showed that a reconnected edge also 
represents a minimum $s'$-$t'$-cut for the new vertices~$s'$
 and $t'$ incident to the edge after the reconnection. Furthermore, the constructed 
 cuts need to be non-crossing in order to be representable by a tree. 
While Gomory and Hu prevent crossings with the help of contractions, 
Gusfield shows that a crossing of an arbitrary 
minimum $s$-$t$-cut with another minimum separating cut can 
be easily resolved, if the latter does not separate $s$ and $t$.
Hence, the cut tree algorithm basically admits the use of arbitrary 
minimum cuts. 

For our special cut tree we choose the following community cuts:
for a vertex pair $\{s,t\}$ let $(S,V\setminus S)$ denote the 
community cut inducing $S:=\SC(s,t)$ and
let $(T,V\setminus T)$ denote the community cut inducing 
$T:=\SC(t,s)$. If $|S|\leq |T|$, we choose $(S,V\setminus S)$, 
and $(T,V\setminus T)$ otherwise. Furthermore, we direct 
the corresponding tree edge to the chosen SC, and we
associate the opposite SC, which was not chosen, also with the edge,
storing it elsewhere for further use.
In Appendix~\ref{app:cutTree}
we show that the so chosen "smallest" community cuts are already non-crossing, hence a transformation according to Gusfield is not necessary. This guarantees that the cuts  represented in the final tree are the same community cuts as chosen for the construction. We further show that after reconnecting an edge, the corresponding cut still induces a "smallest" SC for the vertex the edge points to. Altogether, this proves the following.
\newcommand{\corCutTree}{
For an undirected, weighted graph $G=(V,E,c)$ there
exists a rooted cut tree $\T(G) = (V,E_\T,c_\T)$ with edges directed to the leaves 
such that each edge $(t,s) \in E_\T$ represents $\SC(s,t)$, and $|\SC(s,t)|\leq |\SC(t,s)|$. Such a tree can be constructed by $n-1$ maximum flow\footnote{Max-flows are necessary in order to determine a smallest SC. For general cut trees preflows (after the first phase of common max-flow-push-relabel algorithms) suffice.} computations.
}
\begin{theorem}
\label{cor:cutTree}
\corCutTree
\end{theorem}
At the price of $O(n^2)$ 
additional space, the opposite SCs resulting from the cut tree
construction can be naively stored in an $(n-1)\times n$ matrix,
which admits to check the membership of a vertex to an 
opposite SC in constant time.
In many cases we even need only $k\leq (n-1)$ rows in the matrix,
since some edges share the same SC, and we can 
deduce these edges during the cut tree construction.
However, for few edges the determined opposite SC might become 
invalid again, due to a special situation while reconnecting the edge. 
For these edges we need to recalculate the opposite SCs in a second step. 
Hence, the construction of $\T(G)$ together with the opposite 
SCs associated with the edges in $\T(G)$ can be done by at most
$2(n-1)$ max-flow computations.
We now show that each SC in $M(G)$ is either given by an edge or is an opposite SC associated with an edge in $\T(G)$.

\newcommand{\lemNumMaxSCs}{
For an undirected weighted graph $G= (V,E,c)$ 
it is $n \leq |M(G)| \leq 2(n-1)$. Constructing $M(G)$  
needs at most $2(n-1)$ max-flow computations.
}

\begin{theorem}
  \label{lem:numMaxSCs}
  \lemNumMaxSCs
\end{theorem}
%
\begin{proof}
The SC-tree already represents $n-1$ different maximal source 
communities of~$G$ and there is at least one maximal source 
community of the root that is not represented by the tree.
Hence, there are at least $n$ maximal source communities in $G$.

In order to prove the upper bound we observe the following.
From the structure of the SC-tree if follows that 
if $p$ is a predecessor of $q$, the source community $Q(q,p)$ is 
given by the cheapest edge on the path between $p$ and $q$ that is closest to $q$.
We further show that 
(i) the source community $P(p,q)$ is the opposite source community associated
with the cheapest edge on the path from $p$ to $q$ that is closest to $p$.
If $u$ and $v$ are vertices in disjoint subtrees with $r$ the nearest 
common predecessor, we prove that the source community $U(u,v)$ 
(ii) equals the source community $U'(u,r)$ if no edge on the path from $r$ 
to $v$ is cheaper than the cheapest edge on the path from $r$ to $u$, and 
(iii) equals the source community $R(r,v)$, otherwise.
Since $r$ is a predecessor of $u$ and $v$, together with (i) this finally 
proves that there are at most $2(n-1)$ different maximal source communities in $G$.

\emph{Proof of (i):} 
Let $(t,s)\in E_\T$ denote the cheapest edge on $\pi(p,q)$ that is closest to~$p$.
Obviously it is $\lambda(p,q) = c_\T(t,s)$.
Since $(t,s)$ is closest to $p$, the community cut inducing $P(p,q)$ 
does not separate $p$ and $t$.
Furthermore, it is $p\in T(t,s)$, otherwise the community cut inducing $T(t,s)$ 
can be bend according to Lemma~\ref{lem:shelteredByPrev} such that it 
induces an edge of cost $\lambda(p,q)$ on $\pi(p,q)$ that is closer to $p$ 
than $(t,s)$. Hence, we have $\{p,t\} \subseteq P(p,q)\cap T(t,s)$, while $p\notin T(t,s)$.

If $s\notin P(p,q)$ we get the situation of Lemma~\ref{lem:intersecBehavior}(2ii), 
which yields $P(p,q) = T(t,s)$.
If $s\in P(p,q)$ we get $T(t,s)\subseteq P(p,q)$, according to 
Lemma~\ref{lem:intersecBehavior}(2i). However, since $T(t,s)$ also 
separates $q$ and $p$, it must hold $|P(p,q)| = |T(t,s)|$, which 
contradicts the assumption $s\in P(p,q)$.

\emph{Proof of (ii)}:
If no edge on $\pi(r,v)$ is cheaper than the cheapest edge on $\pi(r,u)$, 
any cheapest edge on $\pi(r,u)$ also induces a minimum $u$-$v$-cut, in 
particular the community cut of $U'(u,r)$ is a minimum $u$-$v$-cut. 
We show now that the community cut inducing $U(u,v)$ is also a 
minimum $u$-$r$, i.e., that it separates $u$ and $r$. It follows that $U(u,v) = U'(u,r)$.

Suppose $r\in U(u,v)$. Since $v\notin U'(u,r)$ we get the situation of 
Lemma~\ref{lem:intersecBehavior}(2i) which yields $U'(u,r)\subseteq U(u,v)$. 
However, since $U'(u,r)$ also separates $u$ and $v$ it must hold $|U(u,v)| = |U'(u,r)|$, 
which contradicts the assumption $r\in U(u,v)$.

\emph{Proof of (iii)}:
If all edges on the path from $r$ to $u$ are more expansive than 
the cheapest edge on the path from $r$ to $v$, the community 
cut inducing $U(u,v)$ does not separate $u$ and $r$, i.e., it is also 
a minimum $r$-$v$-cut. Hence, it follows that $U(u,v) = R(r,v)$, 
since vice versa $u \in R(r,v)$ due to $\lambda(r,v) = \lambda(u,v)$.
\qed
\end{proof}

After precomputing $M(G)$, which includes the construction 
of $\T(G)$ (we denote this by $M(G)\supset \T(G)$), 
the following tools allow to efficiently analyze the SC structure of $G$ with respect to 
different SCs that are already known, for example,
from the cut clustering algorithm of Flake et al.\ or the set~$M(G)$.
The key is Lemma~\ref{lem:shape}.
It limits the shape of arbitrary SCs to subtrees in~$\T(G)$,
which admits an efficient enumeration of disjoint SCs 
by a depth-first search (DFS), as we will see in the following.
\newcommand{\lemShape}{
The subgraph~$\T[T]$ 
induced by a SC~$T$ 
in~$\T(G)$ is connected.  
}
\begin{lemma}
  \label{lem:shape}
  \lemShape
\end{lemma}
\begin{proof}
If $T$ is represented by an edge in $\T(G)$ the assertion obviously holds.
Hence, assume $T$ is an opposite SC or another arbitrary SC.
In order to prove the connectivity of $\T[T]$,
we first focus on the predecessors of $t$.
Let $p$ denote a predecessor of $t$ with $p\in \T[T]$ 
and $q$ a successor of $p$ on $\pi(p,t)$. We prove that $q\in\T[T]$.
Assume $q\notin \T[T]$. Since $t$ is a successor 
of $q$, $t$ is in the SC $Q(t,q)$. 
According to Lemma~\ref{lem:intersecBehavior}(2i), 
however, it follows that $T\subseteq Q$, which contradicts $p\in T$.

In a second step we consider the remaining vertices.
Let $u$ be a vertex that is no predecessor of $t$.
Let $r$ denote the nearest common predecessor of $u$ and $t$.
We first show, that (i) if $u\in \T[T]$, then $r\in \T[T]$.
Then we suppose there is also a predecessor $p\not= r$ 
of $u$ on $\pi(u,r)$ and prove 
(ii) that if $u\in \T[T]$, then $p\in \T[T]$.
Together with the observation on the predecessors of $t$, 
this ensures the connectivity of $\T[T]$.

\emph{Proof of (i)}:
If $r = t$, we are done.
Assume $r\not= t$ and $r\notin \T[T]$
Since $t$ is a successor of $r$, $t$ is in the SC $Q(t,r)$, 
while $u\notin Q(t,r)$.
According to Lemma~\ref{lem:intersecBehavior}(2i), 
however, it follows that $T\subseteq Q$, which contradicts $u\in T$.

\emph{Proof of (ii)}:
From (i) we already know that $r\in \T[T]$.
Assume $p\notin \T[T]$ and consider the SC $P(p,r)$. Is is $t \notin P$, 
and hence, according to Lemma~\ref{lem:intersecBehavior}(1)
$P$ and $T$ are disjoint, contradicting $u\in T$, since $u\in P$.

If $T$ is maximal and
$u\in \T[T]$ is a successor of~$t$ 
or a successor of a predecessor $p$ of~$t$, $u \notin \pi(p,t)$, 
we can further show that
the subtree rooted in~$u$ is in~$\T[T]$.

This is obviously holds if $T$ is represented by a tree edge.
If $T$ is a maximal opposite SC, let $(t,s)\in E_\T$ denote the 
edge $T$ is associated with. 
Let $u\notin S(s,t)$ denote a successor of $t$ or a successor 
of a predecessor $p$ of $t$ with $u\notin \pi(p,t)$. 
With $u \in \T[T] \equiv T(t,s)$ and  $s\notin U(u,p)$, with $U(u,p)$ 
corresponding to the subtree routed in $u$, we get the situation 
in Lemma~\ref{lem:intersecBehavior}(2), and it follows $U(u,p)\subseteq T(t,s)\equiv \T[T]$.
\qed
\end{proof}

\vspace*{-1.5ex}
\subsubsection{Maximal SC Clustering for one SC.}
Given an arbitrary SC $S$, the first tool returns a 
clustering~$\C(S)$ of $G$ that contains $S$, consists of SCs
and is maximum
in the sense that each clustering
that also consists of $S$ and further SCs is hierarchically nested in $\C(S)$.
This implies that $\C(S)$ is the unique
maximal clustering among all clusterings consisting of~$S$ and further SCs.
We call $\C(S)$ the \emph{maximal SC clustering} for~$S$.

\newcommand{\theoMaxSCClus}{
Let~$S$ denote a SC in $G$. The unique maximal SC clustering for~$S$ can be determined in $O(n)$ time after preprocessing $M(G) \supset \T(G)$.}
\begin{theorem}
  \label{theo:maxSCClus}
  \theoMaxSCClus
\end{theorem}
\newcommand{\lemSource}{
Each SC in $\C(S)\setminus \{S\}$ is a SC with respect to the source of~$S$.
}
%
The maximal SC clustering for $S=: S_0$ can be determined 
by the following construction,
 which directly implies a simple algorithm.
Let $r$ denote the root of $\T(G) =: \T_0$ and $\T[S_0]$ the subtree 
induced by $S_0$ in~$\T_0$ (Lemma~\ref{lem:shape}).
Deleting $\T[S_0]$ decomposes $\T_0$ into connected components, 
each of which representing a SC, apart from the one 
containing $r$ if $r\notin S_0$. 
If $r\in S_0$, we are done. Otherwise, let $\T_1$ denote the component containing $r$
and $r_0$ the root of $\T[S_0]$. Obviously is $p_0\in \T_1$ for $(p_0,r_0)\in E_\T$ and
$\SC(p_0,r_0) =: S_1$
induces a subtree $\T[S_1]$ in $\T_1$.
Thus, $S_1$ and $\T_1$ adopt the roles of~$S_0$ and~$\T_0$.

Continuing in this way, we finally end up with a SC $S_k$ containing $r$, such that 
deleting $\T[S_k]$ yields only SCs.
The resulting clustering $\C(S)$ 
consists of $S=S_0$, $S_i$, $i = 1, \dots, k$, 
and the remaining SCs resulting from the decompositions of $\T_0, \dots, \T_k$.

The proof of the maximality of $\C(S)$ is based on the following lemma.

\begin{lemma}
  \label{lem:Source}
  \lemSource
\end{lemma}
\begin{proof}
Let~$c$ denote the source of a SC $C\in \C(S)$ 
and~$s$ the source of~$S$.
Recall that~$C$ is a maximal SC due to the construction of~$\C(S)$.
 Then, $C' :=\SC(c,s)$ 
and~$S$ are disjoint
according to Lemma~\ref{lem:intersecBehavior}(1), 
since $\{c,s\}\cap (C'\cap S) = \emptyset$.

If $C$ is a SC with respect to a vertex $v\in S$ we 
get $C = C'$ according to Lemma~\ref{lem:intersecBehavior}(2ii).

If $C$ is a SC with respect to a vertex $v\notin S$,
let $S'$ denote the cluster containing~$v$. With the 
same arguments as before, $C$ is a SC with respect 
to the source $s'$ of $S'$.
By induction and due to the construction,~$S'$
is a SC with respect to $s$
and $s'$ is on the path between~$c$ and~$s$ in $\T(G)$. 
If  $C'$ contained $s'$, then
the edge in $\T(G)$ indicating $C'$ would also indicate 
the SC of $s'$ with respect to $s$, which is $S'$. 
This contradicts the fact that $c \notin S'$.
Hence, $C'$ does not contain $s'$ and again 
by Lemma~\ref{lem:intersecBehavior}(2ii) it is $C' = C$.
\qed
\end{proof}
Let $Q$ denote an arbitrary SC with source $q$ that does not intersect $S$, 
let $s$ denote the source of $S$,
and let~$C$ denote the SC in $\C(S)\setminus \{S\}$ with $q\in C$.
Since~$C$ is a SC with respect to $s\notin Q$ (Lemma~\ref{lem:Source}) and $q\in Q\cap C$,
it is $Q\subseteq C$, according to Lemma~\ref{lem:intersecBehavior}(2i).
Thus, each SC not intersecting $S$ is nested in a cluster in $\C(S)$.

For the running time we assume that $S$ is given in a structure 
that allows to check the membership of a vertex in time $O(1)$.
Then identifying all clusters in $\C(S)$ (which are subtrees) 
by applying a DFS\footnote{This induces a rooted subtree 
independent from the orientation in $\T(G)$.} starting from the first 
vertex found in each cluster 
can be done in $O(n)$ time,
since checking if a visited vertex is still in $S_i$ takes constant time for
$i = 1, \dots, k$ (recall, that we store the opposite SCs in a matrix).
The remaining subtrees share their leaves with $\T(G)$.

\vspace*{-3ex}
\subsubsection{Overlay Clustering for $k$ disjoint SCs.}
Given $k$ disjoint arbitrary SCs $S_1, \dots, S_k$,
the second tool returns a clustering $\C(S_1,\dots, S_k)$ of~$G$ that
contains $S_1,\dots S_k$, is nested in each maximal SC clustering $\C(S_1), \dots, \C(S_k)$
and is maximum in the sense that each clustering that consists of SCs and
also contains $S_1,\dots, S_k$ is hierarchically nested in $\C(S_1,\dots, S_k)$.
Basically, according to the construction described below,
$\C(S_1,\dots, S_k)$ is
the unique maximal clustering among all clusterings that are nested in the maximal
SC clusterings $\C(S_1), \dots, \C(S_k)$. 
The further properties result
from the maximality of the SC clusterings, as
for each $\C(S_j)$ and each arbitrary SC $S$ that does not intersect
$S_1, \dots, S_k$ (or equals a given SC) there exists a cluster
$C\in \C(S_j)$ with $S\subseteq C$.
Note that the clusters in $\C(S_1,\dots S_k)\setminus \{S_1,\dots, S_k\}$
are not necessarily SCs.
We call $\C(S_1,\dots, S_k)$ the \emph{overlay clustering} for $S_1,\dots S_k$.

\newcommand{\theoMaxSCSuperimp}{
Let~$S_1, \dots, S_k$ denote disjoint SCs in $G$. The unique 
overlay clustering for $S_1,\dots,S_k$
can be determined in $O(kn)$ time after 
preprocessing $M \supset \T(G)$.}
\begin{theorem}
  \label{theo:maxSCSuperimp}
  \theoMaxSCSuperimp
\end{theorem}
%
The overlay clustering for $S_1,\dots, S_k$ can be determined 
by the following inductive construction, which directly implies a simple algorithm.
We first compute the maximal SC clustering $\C(S_1)$ and color the 
vertices in each cluster, using different colors for different clusters.
Now consider the overlay clustering $\C(S_1,\dots, S_i)$ for the first $i$ maximal 
SC clusterings and color the vertices in $S_{i+1}$, which is nested in a cluster of $\C(S_1,\dots, S_i)$,
with a new color.
During the computation of $\C(S_{i+1})$, we then construct 
the intersections of each newly found cluster $C$ with 
the clusters in $\C(S_1,\dots, S_i)$.
To this end we exploit that the intersection of two subtrees in a tree is again
a subtree. 
Hence, the clusters in $\C(S_1,\dots, S_i,S_{i+1})$ will be subtrees in $\T(G)$, 
since the clusters in $\C(S_1),\dots, \C(S_i)$ and $\C(S_{i+1})$ are 
subtrees in $\T(G)$ by Lemma~\ref{lem:shape}.

Let $r'$ denote the first vertex found in $C$ during the computation of $\C(S_{i+1})$.
We mark $r'$ as root of a new cluster in $\C(S_1,\dots, S_i,S_{i+1})$ 
and choose a new color~$x$ for~$r'$, 
besides the color it already has in $\C(S_1,\dots, S_i)$.
When constructing~$C$ (by applying a DFS), we assign 
the current color~$x$ to all vertices 
visited by the DFS as long as the underlying color 
in $\C(S_1,\dots, S_i)$ does not change. 
Whenever the DFS visits a vertex $r''$ (still in C) with a new underlying 
color,
we chose a new color~$y$ for~$r''$ and 
mark $r''$ as root of a subtree of a new cluster in $\C(S_1,\dots, S_i,S_{i+1})$. When the DFS passes $r''$ 
on the way back to the parent\footnote{The predecessor adjacent to $r''$ in the rooted subtree induced by the DFS.} $p$ of $r''$, the color of $p$ in $\C(S_1,\dots, S_i,S_{i+1})$ becomes the current color again.
Continuing in this way yields a coloring that indicates the intersections of~$C$ with $\C(S_1,\dots, S_i)$.
Repeating this procedure for all clusters in $\C(S_{i+1})$ finally yields $\C(S_1,\dots,S_{i+1})$.
The running time is in $O(kn)$, since we just apply~$k$ computations of maximal SC clusterings.

\vspace{-3ex}
\subsubsection{Example.}
We extract two of the many faces of the SC structure of the weighted co-appearance network (called "lesmis") of the characters in the novel Les Miserables~\cite{k-tsgb-93}. 
Figure~\ref{fig:tree} shows the cut tree $\T(\text{"lesmis"})$, the root~$r$ is depicted as filled square.
Figure~\ref{fig:maxClus} shows the maximal SC clustering $\C(R_1)$ for the SC $R_1$ (filled vertices in squared box).
The subtree $\T[R_1]$ induced by $R_1$ in $\T(\text{"lesmis"})$ is indicated by filled vertices in Figure~\ref{fig:tree}.
Since $r\in R_1$, deleting $\T[R_1]$ immediately decomposes $\T(\text{"lesmis"})$ into the unframed singleton SCs and the round framed SCs shown in Figure~\ref{fig:maxClus}. 
The SC~$R_1$ is the larger of the only two non-singleton clusters in the best cut clustering (with respect to \emph{modularity}~\cite{ng-fecsn-04}) found by the cut clustering algorithm of Flake et al. 
On the other hand, $R_1$ is
the smallest  reasonable SC that was found by the cut clustering algorithm containing~$r$. The next smaller SC in the hierarchy that contains~$r$ consists of only three vertices.
The second non-singleton cluster besides~$R_1$ in the best cut clustering is also in~$\C(R_1)$, namely~$A$. Nevertheless, $\C(R_1)$ is not nested in any clustering of the hierarchy. This is, we found a new clustering that contains all non-singleton clusters of the best cut clustering but far less unclustered vertices.
Due to the maximality of $\C(R_1)$, there is also no clustering with less singletons that consists of SCs and contains~$R_1$.

Figure~\ref{fig:overlay} shows the overlay clustering $\C(S_1,\dots, S_6,R_2)$ with 
$S_1,\dots, S_6$ defined by the non-singleton subtrees of $r$ in $\T(\text{"lesmis"})$. The SC $R_2$ (filled vertices in squared box) has been computed additionally. It equals $\SC(r,T)$ with $T:= \bigcup_{i= 1}^{6} S_i$. 
If we consider the filled vertices in Figure~\ref{fig:overlay} as one cluster~$F := V\setminus T$, then $S_1,\dots, S_6$ together with $F$ represent the overlay clustering $\C(S_1,\dots, S_6)$. However, $\C(S_1,\dots, S_6)$ does not only consist of SCs since $F$ is no SC:
Observe that for the two vertices $v_1,v_2 \in F\setminus R_2$ there exists a vertex $u\in T$ (unfilled square) such that $\SC(v_i,u)\subseteq F$ ($i = 1,2$) is a singleton. Hence, according to Lemma~\ref{lem:intersecBehavior}(2i), any SC in~$F$, apart from $\{v_1\}$ and $\{v_2\}$, must be in $R_2$.  
This is, in contrast to $\C(S_1,\dots, S_6)$, the overlay clustering $\C(S_1,\dots, S_6,R_2)$ consists of SCs and
any clustering that also consists of SCs and contains~$S_1,\dots, S_6$ is nested in~$\C(S_1,\dots, S_6,R_2)$.

\begin{figure}[bt]
\centering
\subfigure[Basic cut tree $\T(G)$]{
		\label{fig:tree}
		\includegraphics[width = 4.1cm, page=1]{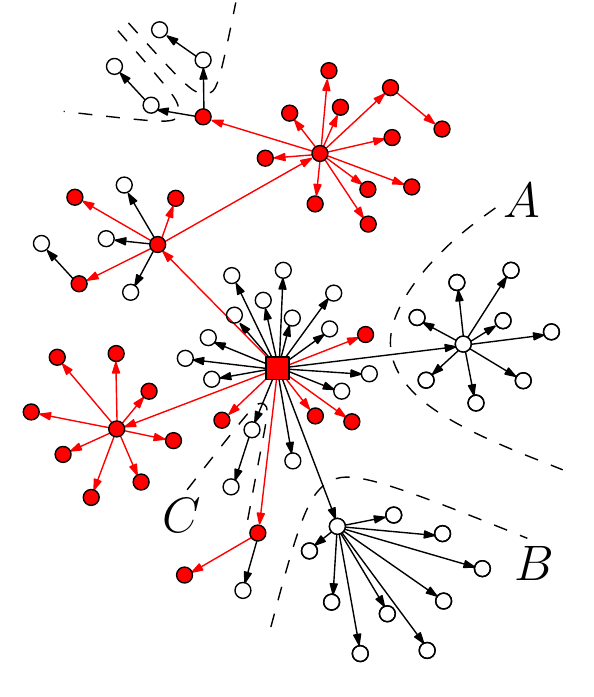}
	}
\hspace{-4ex}
	\subfigure[Maximal SC clustering]{
		\label{fig:maxClus}
		\includegraphics[width = 3.9cm, page=2]{lesmis_single.pdf}
	}
\hspace{-2ex}
	\subfigure[Overlay clustering]{
		\label{fig:overlay}
		\includegraphics[width = 3.9cm, page=3]{lesmis_single.pdf}
	}
	\vspace*{-2ex}
	\caption{Exemplary clusterings of the lesmis-network; $A,B,C$ appear in both clusterings.}
	\label{fig:example}
\vspace{-0ex}
\end{figure}

\vspace{-1ex}
\section{Conclusion}
\vspace{-1ex}
Based on minimum separating cuts and maximum flows, respectively, we characterized SCs, a special type of predominantly connected communities. We introduced a method for efficiently computing a complete hierarchy of clusterings consisting of SCs according to Flake et al.~\cite{ftt-gcmct-04}.
Furthermore, we exploited the structure of cut trees~\cite{gh-mtnf-61} in order to develop a framework that admits the efficient construction of maximal SC clusterings and overlay clusterings for given SCs, after precomputing at most $2(n-1)$ maximum flows.
In most cases, however, we expect only around $n-1$ maximum flows for the preprocessing, since the cases that cause the additional flow computations (when the opposite SC becomes invalid during the construction of the cut tree) are rare in practice. For the "lesmis" network in the previous example we needed only $n+3$ maximum flows with $n=77$.
We remark that a single maximal SC clustering for $S$ can be also constructed directly by iteratively computing maximal SCs of the vertices not in $S$ with respect to the source of $S$. However, in the worst case, this needs $|V\setminus S|$ flow computations, if the SCs are singletons or if they are considered in an order that causes many unnecessary computations of nested SCs.
In contrast, due to its short query times, our framework efficiently supports the detailed analysis of a networks's SC structure 
with respect to many different maximal SC clusterings and overlay clusterings.

\bibliographystyle{plain}

 \bibliography{predomCom_full}

\newpage
\begin{appendix}

\begin{bibunit}

\chapter*{Appendix}

\section{Proof of Theorem~\ref{theo:main} and Experimental Evaluation}
\label{app:CutClus}

Our simple approach for constructing a complete hierarchy of cut clusterings exploits the properties of cut-cost functions. The \emph{cut-cost function} $\omega_S$ of a set $S \subseteq V$ is a linear function in $\alpha$ that represents the costs of cut $(S, V_\alpha \setminus S)$ in $G_\alpha$ based on the costs of cut $(S, V\setminus S)$ in $G$ and the size of $S$.
\begin{eqnarray}
\omega_S : \mathbb R^+_0 & \longrightarrow & [c(S,V \setminus S), \infty) \subset \mathbb R^+_0 \nonumber \\
\omega_S (\alpha) & := & c(S,V \setminus S) + |S|\;\alpha \nonumber
\end{eqnarray}
In the context of clusters in a cut clustering hierarchy we call a cluster $C'$ a \emph{child} of a cluster $C$ if $C' \subset C$. The cluster $C$ is a \emph{parent} of $C'$. Hence, a cluster might have several children and several parents.
Let $C'$ denote a child of a cluster $C$, i.e., $|C'| < |C|$. Then, the slope of $\omega_{C'}$, which is given by $|C'|$, is positive but less than the slope of $\omega_C$ (cp. Figure~\ref{fig:cutfunc_app}). If $\omega_C$ and $\omega_{C'}$ intersect, let $\alpha^*$ denote the intersection point. Note that the functions of a parent and a child do not intersect in general.
For each $\alpha \in [0,\alpha^*)$ it is then $\omega_C(\alpha) < \omega_{C'}(\alpha)$, and we say that the parent $C$ \emph{dominates} the child $C'$. This is, $C'$ will never become a cluster in $G_{\alpha}$, as $C$ induces a smaller $u$-$t$-cut for each $u \in C'$, which prevents $C'$ from becoming a community of any vertex.  

\begin{wrapfigure}[8]{r}{.3\textwidth}
\vspace{-7ex}
\centering
\includegraphics[width = 2.5cm]{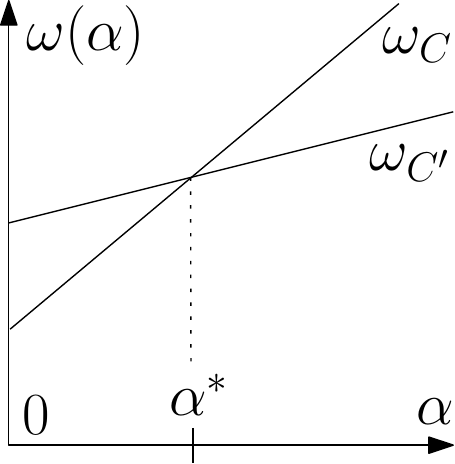}
\vspace{-1.5ex}
\caption{Intersecting cut-cost functions. }
\label{fig:cutfunc_app}
\end{wrapfigure}
If the child $C'$ contains the representative $r(C)$ of the parent $C$, we further observe that the child $C'$ \emph{dominates} the parent $C$ \emph{with respect to $r(C)$} for each $\alpha \in [\alpha^*,\infty)$. This is, $C$ will never become a cluster with representative $r(C)$ in $G_{\alpha}$, as $\omega_{C'}(\alpha) \leq\omega_{C}(\alpha)$ for each $\alpha \in (\alpha^*,\infty)$ and $|C'| < |C|$. Thus, $C'$ either induces a smaller $r(C)$-$t$-cut or, if the cut costs are equal, $C'$ induces a smaller cut side, which both prevents $C$ from becoming a community of $r(C)$.

The latter observation implies that the function of a cluster $C$ with representative $r(C)$ always intersects with the function of any child $C'$ of $C$ with $r(C)\in C'$. Otherwise, $C'$ would dominate $C$ wrt.\ $r(C)$ in the whole parameter range contradicting the fact that $C$ is a cluster with representative $r(C)$ in the hierarchy.

For two consecutive hierarchy levels $\C_i < \C_{i+1}$ we call $\alpha'$ the \emph{breakpoint} if CutC returns $\C_i$ for $\alpha'$ and $\C_{i+1}$ for $\alpha' - \epsilon$ with $\epsilon \rightarrow 0$.
A breakpoint between two hierarchy levels is in particular an intersection point of the cut-cost functions of two clusters $C' \subset C$. The simple idea of our parametric search approach is to compute relevant intersection points and check if they yield new clusterings. 

\vspace{2ex}
\rephrase{Theorem}{\ref{theo:main}}{\paraApp}

\begin{proof}
This proof constructively describes the steps of our parametric search approach and shows the correctness. The first step is the construction of $\alpha_m$. 
Formally, we define
$\alpha_m := \min_{C\in\C_j }\lambda_{C}$ with $\lambda_{C} := \max_{C'\in \C_i: C'\subset C}\{\alpha\mid \omega_{C}(\alpha) = \omega_{C'}(\alpha) \} $.
The notation $\lambda_{C}$ describes the maximum intersection point of the function $\omega_C$, $C\in \C_j$,
with the functions of all children of $C$ on level $\C_i$. The minimum of these points 
then yields $\alpha_m$. Note that $\alpha_m$ is well-defined as each parent function intersects with at least one child function.
In practice we construct $\alpha_m$
by iterating the list of representatives stored for
$\C_i$. Since $\C_j$ assigns each of these representatives to a cluster, matching the children to their parents can be done in time $O(|\C_i|)$. This already dominates the costs for the remaining steps, which are the computation of the intersection points and the search for the maximum and minimum values.
Recall that the clusters in both clusterings are also mapped to their sizes and costs, which allows to compute an intersection point in constant time.

In the following let $C^{}\in \C_j$ denote a parent that induces $\alpha_m$, i.e., $\lambda_{C^{}} = \alpha_m$.
Furthermore, let  $C^{_1} \in \C_i$ denote a child of $C^{}$ that contains the representative $r(C^{})$ and let $C^{_2}\in \C_i$ denote a child of $C^{}$ with $\omega_{C^{_2}}(\alpha_m) = \omega_{C^{}}(\alpha_m)$. Thus, the intersection point for $C^{_2}$ and $C$ is $\alpha_m$. We denote the intersection point for $C^{_1}$ and $C$, which also exists, by $\alpha^{_1}$. 

\textit{Claim 1: $\alpha_j < \alpha_m \leq \alpha_i$}.
Suppose first $\alpha_j \geq \alpha_m$. This implies $\alpha_j \in [\alpha^{_1},\infty)$, and thus, $C^{_1}$ would dominate $C$ wrt.\ $r(C)$ in $G_{\alpha_j}$ contradicting the fact that $C$ is a cluster in $\C_j$ with representative $r(C)$.
Suppose secondly $\alpha_i < \alpha_m$. This implies $\alpha_i \in [0,\alpha_m)$, and thus, $C$ would dominate $C^{_2}$ in $G_{\alpha_i}$ contradicting the fact that $C^{_2}$ is a cluster in $\C_i$.

After having computed $\alpha_m$ we apply CutC with this newly obtained value. The resulting clustering is denoted by $\C_m$. According to Claim~1 and the hierarchical structure it is $\C_i \leq \C_m \leq \C_j$.

\textit{Claim 2: $\C_m \not= \C_j$}.
Recall that $\alpha^{_1} \leq \alpha_m$. This implies $\alpha_m \in [\alpha^{_1},\infty)$, and thus, $C^{_1}$ dominates $C$ wrt.\ $r(C)$ in $G_{\alpha_m}$.
Nevertheless, $C^{}$ might be a cluster in $\C_m$ wrt.\ to another representative $u \not= r(C^{})$.
However, this can be disproven by the same argument, since the intersection point for $C^{}$  and the child containing $u$ is, analogously to $\alpha^{_1}$, also at most $\alpha_m$. Recall that $\alpha_m$ is the maximum intersection point regarding the children of $C$.
Thus, $\C_j$ contains at least one cluster $C^{} \notin \C_m$.

\textit{Claim 3: If $\C_m = \C_i$ then $\alpha_m$ is the breakpoint between $\C_i$ and $\C_j$}.
We first show that $\alpha_m$ is the breakpoint between $\C_i$ and the next higher level in the complete hierarchy. In a second step we prove that the clustering on the next higher level equals $\C_j$.
To see the first assertion consider $\alpha_m - \varepsilon < \alpha_m$ for $\varepsilon \rightarrow 0$. This implies $\alpha_m - \varepsilon \in [0,\alpha_m)$, and thus, $C$ dominates $C^{_2}$ in $G_{\alpha_m - \varepsilon}$.
Consequently, $\C_m = \C_i$ contains a cluster $C^{_2}$ that will never appear for $\alpha_m - \epsilon$, and thus, $\alpha_m$ is the breakpoint between $\C_i$ and the next higher level in the complete hierarchy.
In order to prove the latter assertion saying that the next higher level equals $\C_j$, we show that each cluster in~$\C_j$ corresponds to a community in $G_{\alpha_m - \varepsilon}$. The nesting property for communities together with the hierarchical structure then ensures that CutC returns $\C_j$ for $\alpha_m-\varepsilon$, which means that there exists no further clustering between $\C_i$ and $\C_j$.
For this final step we overload the notation of $C$ and $C^{_2}$ as follows: Let $C\in \C_j$ denote an arbitrary cluster and let $C^{_2}\in \C_i = \C_m$ denote a child of $C$ with $\omega_{C^{_2}}(\lambda_C) = \omega_C(\lambda_C)$. This is, the intersection point for $C$ and $C^{_2}$ is $\lambda_C \geq \alpha_m$.
Let further $r$ denote the representative of $C^{_2}$ in $\C_m = \C_i$. Recall that $\C_m = \C_i$ does not imply the equivalence of the representative of $C^{_2}$ in $\C_i$ and the representative of $C^{_2}$ in $\C_m$. We show that (a) $\lambda_C = \alpha_m$, and based on this, that (b) $C$ equals the community of $r$ in $G_{\alpha_m - \epsilon}$.

Sub-claim (a): $\lambda_C = \alpha_m$. Suppose $\lambda_{C} > \alpha_m$, which implies $\alpha_m \in [0, \lambda_C)$. Then the parent $C$ would dominate the child $C^{_2}$ in $G_{\alpha_m}$ contradicting the fact that $C^{_2}$ is a cluster in $\C_m$.

Sub-claim (b): $C$ equals the community of $r$ in $G_{\alpha_m - \epsilon}$.
Let $\bar C$ denote the community of $r$ in $G_{\alpha_m - \epsilon}$. 
Then the hierarchical structure implies $C^{_2}\subseteq \bar C \subseteq C$.
It is further $\omega_{C^{_2}}(\alpha_m) \leq \omega_{\bar C}(\alpha_m)$, as otherwise $\bar C$ would induce a smaller $r$-$t$-cut in $G_{\alpha_m}$ than the actual community $C^{_2}$ of $r$. On the other hand, it is $\omega_{\bar C}(\alpha_m-\varepsilon) \leq \omega_{C}(\alpha_m-\varepsilon)$, by the same argument, i.e., otherwise $C$ would induce a smaller $r$-$t$-cut in $G_{\alpha_m -\varepsilon}$ than the actual community $\bar C$ of $r$. 
With the help of (a) we see that $\omega_{C^{_2}}(\alpha_m) = \omega_{C}(\alpha_m)$, and thus, the cost-function of $\bar C$ must lie above $\omega_C$ in $\alpha_m$ and below $\omega_C$ in $\alpha_m-\varepsilon$ (cp.\ Figure~\ref{fig:cutfunc_2}). 
This implies that the slope of $\omega_{\bar C}$ is at least the slope of $\omega_C$.
Now suppose $\bar C \not= C$, which implies $\bar C \subset C$ and thus $|\bar C| < |C|$. The latter, however, means that the slop of $\omega_{\bar C}$ is less than the slope of $\omega_C$ contradicting the previous observation. Hence, it is $\bar C = C$.

\qed
\end{proof}

\begin{wrapfigure}[8]{r}{.3\textwidth}
\vspace{-8ex}
\centering
\includegraphics[width = 2.5cm]{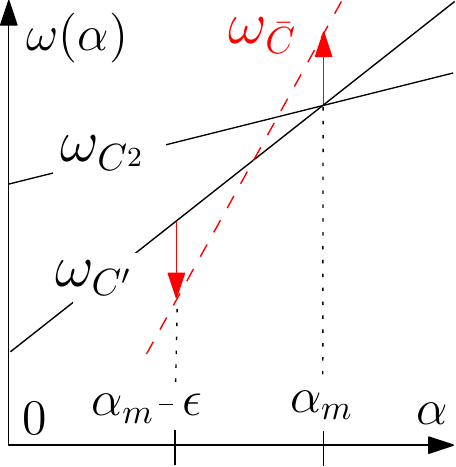}
\vspace{-1.5ex}
\caption{Intersecting cut-cost functions. }
\label{fig:cutfunc_2}
\end{wrapfigure}
Theorem~\ref{theo:main} allows
a recursive search beginning with the computation of $\alpha_m$ for the trivial clusterings $\C_0 < \C_{\max}$ ($\alpha_0 > \alpha_{\max} $).
Recall that in $\C_0$ each vertex is a cluster; $\C_{\max}$ consists of the set of connected components.
After applying CutC for $\alpha_m$ the resulting clustering can be easily compared to the current lower level by counting clusters. If a new clustering was found, the recursion branches and the list storing the levels of the hierarchy is updated. Otherwise, the current branch stops since the breakpoint between two consecutive clusterings has been found.
In contrast to a binary search on the discretized parameter range this approach 
definitely returns a complete hierarchy.

\subsubsection{Experimental Evaluation.}
For our experiments we used real world instances as well as generated instances. Most instances are taken from the testbed of the 10th DIMACS Implementation Challenge~\cite{-dimac-11}, which provides benchmark instances for partitioning and clustering. 
The implementation was realized within the LEMON framework~\cite{-lemon-11}, version~1.2.1. We implemented CutC as described in Algorithm~\ref{alg:CutC},   extended by a heuristic that chooses the vertices in non-increasing order w.r.t.\ the weighted degree. Due to this heuristic, which was proposed by Flake et al., the number of min-cut computations in CutC becomes proportional to the number of clusters in the resulting clustering~\cite{ftt-gcmct-04}.
The min-cut implementation provided by LEMON runs in $O(n^2 \sqrt{m})$. Note that we did not focus on a notably fast implementation. Instead, the implementation should be simple and practical using available routines for sophisticated parts like the min-cut computation.
Table~\ref{tab:runtime1} lists ascending CPU times determined on an AMD Opteron Processor 252 with 2.6 GHz and 16 GB RAM. 

For comparison, we further ran a binary search on the same instances, using the same CutC implementation in the same framework. The running times are listed twice in Table~\ref{tab:runtime1}, once as CPU times and again as factors saying how much longer the binary search ran compared to the parametric search.
However, this is not meant to be a competitive running time experiment, since the running time of the binary search mainly depends on the discretization. We just want to demonstrate that being compelled to choose the discretization intuitively, without any knowledge on the final hierarchy, makes the binary search less practical.
From a users point of view focussing on completeness, we defined the size of the discretization steps as $1/n^2$. The dependency on $n$ is motivated by the fact that the potential number of levels increases with $n$, and by the hope that the breakpoints are distributed more or less equidistantly.
For small graphs with $n<1000$, one can even afford some more running time. Thus, we reduced the step size for those graphs to $1/(1000\; n)$ in further support of completeness. This yields $2^{10}$ to $2^{30}$ discretization steps depending on the length of the parameter range $[\alpha_{\max},\alpha_0]$.
With this discretization the binary search exceeds the parametric search by a factor of four up to~$32$. 

Furthermore, as expected, the running time does not only depend on the input size
but also on the number of different levels in the hierarchy. This can be observed for both approaches comparing the instances as-22july06 and cond-mat. Although the latter is smaller, it takes longer to compute 80 levels compared to only 33 levels in the former graph.

\renewcommand{\arraystretch}{1}
\setlength{\tabcolsep}{1.4mm}
\begin{table}[h]
\centering
\caption{Running times for the parametric search approach (PasS) and the binary search approach (BinS) in minutes and seconds. The factors listed for BinS describe how much longer BinS ran compared to ParS. Instances are sorted by CPU times for ParS. Times longer than six days are marked by *.}

\begin{tabular}{l | c | c | c | r| r| r}
graph & n & m & h & ParS~[m:s] & BinS~[m:s] & BinS~[fac]
\\
\hline
celegans\_metabolic	 & 453 & 2025 & 8 & 0.300 & 7.620 & 8.380\\
celegansneural &	297&	2148& 	17	&0.406 &	8.653	& 9.919Ê\\
netscience&	1589&	2742&	38&	4.310&	4.030&	11.952 \\
power&	4941&	6594&	66&	1:25.736&	8.773&	15.742 \\
as-22july06&	22963	&48436&	33&	39:54.495&	12.419	&20.583 \\
cond-mat&	16726&	47594	&80&	44:15.317&	14.917&	27.425 \\
rgg\_n\_2\_15&	32768&	160240&	46	&245:25.644&	32.748&	22.573\\
G\_n\_pin\_pout&	100000&	501198&	4	&369:29.033&	*&	* \\
cond-mat-2005& 	40421&	175691&	82&	652:32.163&	*	& 21.446 \\
\end{tabular}
\label{tab:runtime1}
\vspace{-2ex}
\end{table}

\newpage
\section{Proof of Theorem~\ref{cor:cutTree} and the Cut Tree Algorithm}\label{app:cutTree}
In this section we show that applying the cut tree algorithm of Gomory and Hu with smallest community cuts as described in Section~\ref{sec:framework} yields a cut tree as stated in Theorem~\ref{cor:cutTree}.

\vspace{2ex}
\rephrase{Theorem}{\ref{cor:cutTree}}{\corCutTree}

\vspace{2ex}
To this end, we briefly review the cut tree algorithm of Gomory and Hu~\cite{gh-mtnf-61} and prove Lemma~\ref{lem:refineInside} and Lemma~\ref{lem:refineOutside}, which together guarantee the correctness of our construction and show how the opposite SCs can be also retained.
Recall that for our special cut tree we choose the following community cuts in line~\ref{alg:gh:minCut}, Algorithm~\ref{alg:gh}:
for a vertex pair $\{s,t\}$ let $(S,V\setminus S)$ denote the community cut inducing $\SC(s,t)$ and
let $(T,V\setminus T)$ denote the community cut inducing $\SC(t,s)$. If $|\SC(s,t)|\leq |\SC(t,s)|$ we choose $(S,V\setminus S)$, and $(T,V\setminus T)$ otherwise. We call the chosen SC a "smallest" SC with respect to $s$ and $t$ and orientate the resulting edge in the intermediate tree such that is points to the SC.
Furthermore, we choose in line~\ref{alg:gh:uv}, Algorithm~\ref{alg:gh}, the last considered vertex in $S$ together with a new vertex in $S$ as $\{u,v\}$.
\newcommand{\lemRefineInside}{
Let $S$ denote a smallest SC with respect to $s \in S$ and $t \in V\setminus S$ 
and~$S'$ a smallest SC with respect to $s$ and another vertex $x\in S$.
Then $S' \subset S$.  If further $s\in S'$, then $S$ is a smallest SC 
with respect to $x$ and $t$ and $\SC(t,s) = \SC(t,x)$.
}
\begin{lemma}
  \label{lem:refineInside}
  \lemRefineInside
\end{lemma} 

\begin{figure}[h]
\centering
	\subfigure[$x\in S'$]{
		\label{fig:cutPairsNew_a}
		\includegraphics[width = 3.5cm, page=1]{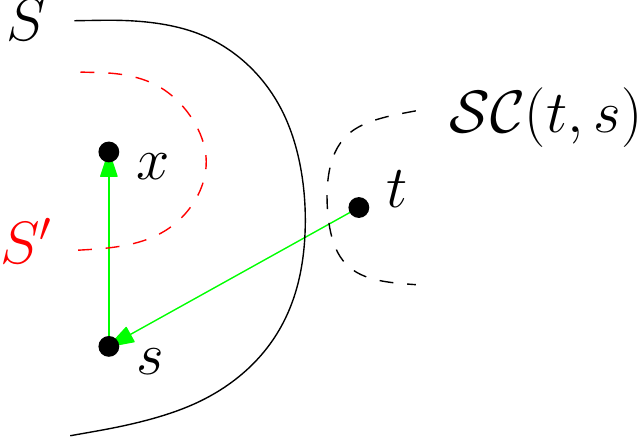}
	}
\hspace{1ex}
	\subfigure[$s\in S'$]{
		\label{fig:cutPairsNew_b}
		\includegraphics[width = 3.5cm, page=2]{refineInside.pdf}
	}
	\vspace*{-2ex}
	\caption{Situation in Lemma~\ref{lem:refineInside}.}
	\label{fig:refineInside}
\vspace{-0ex}
\end{figure}

\newcommand{\lemRefineOutside}{
Let $U$ denote a smallest SC with respect to $u \in U$ and $s\in V\setminus U$ 
and let $S'$ denote a smallest SC with respect to $s$ and another vertex $x\in V\setminus U$.
Then $U \subset S'$ or $U\cap S' = \emptyset$.

If $s\in S'$ and $U\cap S' = \emptyset$, then $U$ is also a smallest SC with respect to~$u$ and~$x$ and
if (i) $x\in \SC(s,u)$, $\SC(s,u) = \SC(x,s)$ and otherwise
(ii) $\SC(s,u) = S'$ and $\SC(x,u) = \SC(x,s)$ if also $u\notin SC(x,s)$.

If $x\in S'$ and $U\subset S'$, then $U$ is also a smallest SC with respect to~$u$ and~$x$ and
if (i) $x\in \SC(s,u)$, $\SC(s,u) = \SC(x,u)$ and otherwise 
(ii) $\SC(s,u) = \SC(s,x)$, but no assertion on $\SC(x,u)$, which is the missing opposite SC of the edge $(x,u)$.
}

\begin{lemma}
  \label{lem:refineOutside}
  \lemRefineOutside
\end{lemma} 
\begin{figure}[h]
\centering
	\subfigure[]{
		\label{fig:cutPairsNew_c}
		\includegraphics[width = 2.5cm, page=1]{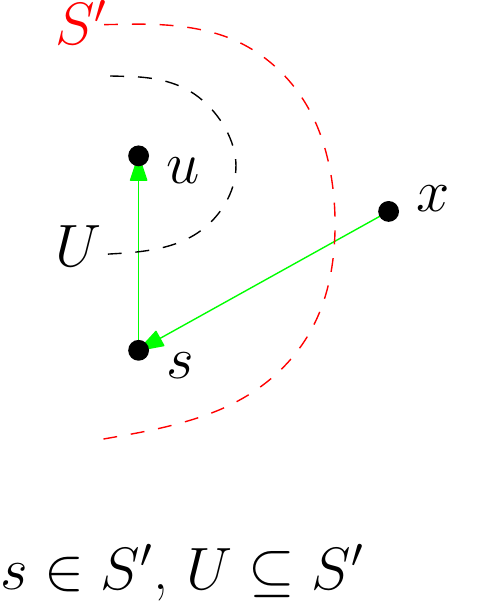}
	}
\hspace{1ex}
	\subfigure[]{
		\label{fig:cutPairsNew_d}
		\includegraphics[width = 2.5cm, page=2]{refineOutside.pdf}
	}
\hspace{1ex}
	\subfigure[]{
		\label{fig:cutPairsNew_e}
		\includegraphics[width = 2.5cm, page=3]{refineOutside.pdf}
	}
\hspace{1ex}
	\subfigure[]{
		\label{fig:cutPairsNew_f}
		\includegraphics[width = 2.5cm, page=4]{refineOutside.pdf}
	}
	\vspace*{-2ex}
	\caption{Situation in Lemma~\ref{lem:refineOutside}.}
	\label{fig:refineOutside}
\vspace{-0ex}
\end{figure}

\subsubsection{Reviewing the Cut Tree Algorithm.}
%
We briefly revisit the construction of a cut tree~\cite{gh-mtnf-61,g-vsmap-90}. 
This algorithm iteratively constructs $n-1$ non-crossing minimum separating cuts for $n-1$ vertex pairs, which we call \emph{step pairs}. These pairs are chosen arbitrarily from the set of pairs not separated by any of the cuts constructed so far. 
Algorithm~\ref{alg:gh} briefly describes the cut tree algorithm of Gomory and Hu.
\begin{algorithm2e}[h]
	\caption{\algo{Cut Tree}}\label{alg:gh}
	\KwIn{Graph $G=(V,E,c)$}
	\KwOut{Cut tree of $G$}
	Initialize tree $\Ts := (\Vs,\Es,\css)$ with $\Vs \gets \{V\}, \Es \gets \emptyset$ and $\css$ empty\nllabel{alg:gh:init}\\ 
	\While({\myco*[f]{unfold all nodes}}){$\exists S \in \Vs$ with $|S| > 1$}{
		$\{u,v\} \gets$ arbitrary pair from $\binom{S}{2}$ \nllabel{alg:gh:uv}\\
		\lForAll{$S_j$ adjacent to $S$ in $\Ts$}{$N_j \gets$ subtree of $S$ in $\Ts$ with $S_j \in N_j$} \nllabel{alg:gh:subtrees}\\
		$G_S = (V_S,E_S,c_S) \gets$ in $G$ contract each $N_j$ to $[N_j]$ \myco*[f]{contraction} \nllabel{alg:gh:contract}\\
 		$(U,V \setminus U) \gets$  min-$u$-$v$-cut in $G_S$, cost $\lambda(u,v)$, $u \in U$ \nllabel{alg:gh:minCut}\\
		$S_u \gets S \cap U$ and $S_v \gets S \cap (V_S \setminus U)$ \myco*[f]{split $S = S_u \cup S_v$}\\
		$\Vs \gets (\Vs \setminus \{S\}) \cup \{S_u,S_v\}$, $\Es \gets \Es \cup \{\{S_u,S_v\}\}$, $\css(S_u,S_v) \gets \lambda(u,v)$\\
		\ForAll{former edges $e_j = \{S,S_j\} \in \Es$
			\nllabel{alg:gh:reconnectStart}}{
			\lIf({\myco*[f]{reconnect $S_j$ to $S_u$}}){$[N_j] \in U$}{
				$e_{j} \gets \{S_u,S_j\}$
			}
			\lElse({\myco*[f]{reconnect $S_j$ to $S_v$}}){
				$e_j \gets \{S_v,S_j\}$ \nllabel{alg:gh:reconnectEnd}
			}
		}
	}
	\Return $\Ts$
\end{algorithm2e}

The \emph{intermediate} cut tree $\Ts = (\Vs,\Es,\css)$
is initialized as an isolated, edgeless node containing all original vertices. 
Then, until each node of~$\Ts$ is a singleton node, a node $S \in V_*$ is \emph{split}.
\begin{figure}[tb]
\centering
	\subfigure[If $x \in S_{u}$, $\{x,y\}$ is still a cut pair of $\{S_{u}, S_j\}$]{
		\label{fig:cutPairs_a}
		\includegraphics[width = 5.6cm, page=1]{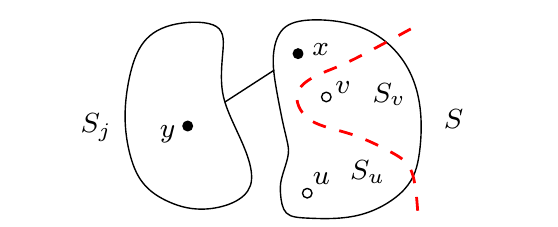}
	}
\hspace{1ex}
	\subfigure[If $x\notin S_{u}$, $\{u,y\}$ is a cut pair of  $\{S_u,S_j\}$]{
		\label{fig:cutPairs_b}
		\includegraphics[width = 5.6cm, page=2]{cutPairs.pdf}
	}
	\vspace*{-2ex}
	\caption{Situation in Lemma~\ref{lem:cut_pairs}. There always exists a cut pair of the edge $\{S_{u}, S_j\}$ in the nodes incident to the edge, independent of the shape of the split cut (dashed).}
	\label{fig:cutPairs}
\vspace{-0ex}
\end{figure}
To this end, nodes $S' \neq S$ are dealt with by contracting in $G$ whole subtrees $N_j$ of $S$ in $\Ts$, connected to $S$ via edges $\{S,S_j\}$, to single nodes $[N_j]$
before cutting, which yields $G_S$.
The split of $S$ into $S_u$ and $S_v$ is then defined by a minimum $u$-$v$-cut (\emph{split cut}) in $G_S$,
which does not cross any of the previously used cuts due to the contraction technique. 

\begin{wrapfigure}[13]{r}{.35\textwidth}
\vspace{-6.5ex}
	\includegraphics[width = 4cm]{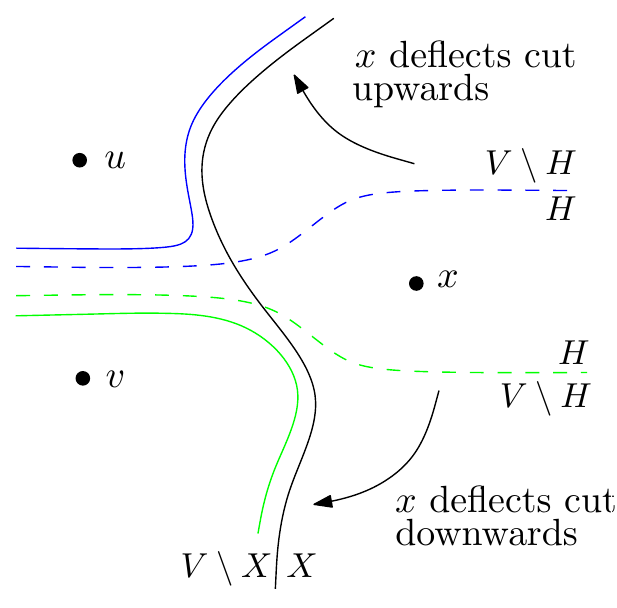}
	\caption{\small Depending on $x$ Lem.~\ref{lem:shelteredByPrev} bends the cut $(H, V \setminus H)$ upwards or downwards.
}
	\label{fig:shelteredByPrev}
\vspace{-3ex}
\end{wrapfigure}
Afterwards, each $N_j$ is reconnected, again by $S_j$,
to either $S_u$ or $S_v$ depending on which side of the cut $[N_j]$ ended up.
Note that this cut in $G_S$ can be proven to induce a minimum $u$-$v$-cut in $G$.
The correctness of \textsc{Cut Tree} is guaranteed by Lemma~\ref{lem:cut_pairs}, 
which takes care for the \emph{cut pairs} of the reconnected edges.
It states that each edge $\{S,S'\}$ in~$\Ts$ has a cut pair $\{x,y\}$ with $x\in S$, $y\in S'$. An intermediate cut tree satisfying this condition is \emph{valid}. The assertion is not obvious, since the nodes incident to the edges in~$T_*$ change whenever the edges are reconnected. Nevertheless, each edge in the final cut tree represents a minimum separating cut of its incident vertices, due to Lemma~\ref{lem:cut_pairs}.
The lemma was formulated and proven in~\cite{gh-mtnf-61} and rephrased in~\cite{g-vsmap-90}. See Figure~\ref{fig:cutPairs}.

\begin{lemma}[Gus.~\cite{g-vsmap-90}, Lem.~4] \label{lem:cut_pairs} 
	Let $\{S,S_j\}$ be an edge in $\Ts$ inducing a cut with cut pair $\{x,y\}$, w.l.o.g.\ $x\in S$.
	Consider step pair $\{u,v\} \subseteq S$ that splits $S$ into $S_u$ and $S_v$, w.l.o.g.\ $S_j$ and $S_u$ ending up on the same cut side, i.e.\ $\{S_{u}, S_j\}$ becomes a new edge in $\Ts$.
	If $x\in S_{u}$, $\{x,y\}$ remains a cut pair for $\{S_{u}, S_j\}$.
	If $x\in S_{v}$, $\{u,y\}$ is also a cut pair of $\{S_{u}, S_j\}$.
\end{lemma}
While Gomory and Hu use contractions in $G$ to prevent crossings of the cuts, as a simplification, Gusfield introduced the following lemma showing that contractions are not necessary, since any arbitrary minimum separating cut can be bent along the previous cuts resolving any potential crossings. See~Figure~\ref{fig:shelteredByPrev}.
\begin{lemma}[Gus.~\cite{g-vsmap-90}, Lem.~1]\label{lem:shelteredByPrev}
	Let $(X,V \setminus X)$ be a minimum $x$-$y$-cut in $G$, with $x \in X$.
	Let $(H, V \setminus H)$ be a minimum $u$-$v$-cut, with $u,v \in V \setminus X$ and $x \in H$.
	Then the cut $(H \cup X, (V \setminus H)\cap(V \setminus X))$ is also a minimum $u$-$v$-cut.
\end{lemma}

\subsubsection{Proof of Lemma~\ref{lem:refineInside} and Lemma~\ref{lem:refineOutside}.}
In the following proofs, whenever we bend a 
cut along another cut deflected by a vertex,
we apply Lemma~\ref{lem:shelteredByPrev}.

\vspace{2ex}
\rephrase{Lemma}{\ref{lem:refineInside}}{\lemRefineInside}

\begin{proof}
We distinguish two cases, namely $x\in S'$ (Figure~\ref{fig:cutPairsNew_a}) 
and $s \in S'$ (Figure~\ref{fig:cutPairsNew_b}).
The cases are (almost) symmetric with respect to the
first assertion. Hence we prove the first
assertion, which is $S'\subset S$, just for the first one. 
Nevertheless, we need to distinguish the second case,
since here the edge $(t,s)$ is reconnected and we need 
to further show that the SCs remain valid. 

\emph{Case 1: $x \in S'$.} 
If $t\notin S'$, it is $S' \subset S$ according to 
Lemma~\ref{lem:intersecBehavior}(2i).

Now suppose $t\in S'$ and consider $\SC(s,x)$. 
Note, that $t \notin \SC(s,x)$, since $\SC(s,x)\cap S' = \emptyset$. 
This is,  $\SC(s,x)\subset S$ due to bending the 
corresponding cut along~$S$ deflected by~$t$. 
We show that $|\SC(s,x)|<|S'|$ and hence, 
$S'$ would not be a smallest SC with 
respect to $s$ and $x$, which is a contradiction. 
Hence, the case $t'\in S'$ does not occur.

Let $\theta$ denote the cut inducing $S$ and $\theta'$ 
the one inducing $S'$. 
We observe that $\theta'$ could be bent along $S$ deflected by $t$, 
and thus, $\theta$ is a minimum $x$-$t$-cut according 
to the correctness of the cut tree algorithm (Lemma~\ref{lem:cut_pairs}).
Hence, $\theta$ could be bent along (the original) $\theta'$ deflected by $s$ yielding a 
cut side $T\ni t$ of a minimum $s$-$t$-cut. 
Since $S$ is a smallest SC with respect to $s$ and $t$,
it follows $|S| \leq |T|$, while $T\subset S'$, $\SC(s,x)\subset S$. 
Finally it is $|\SC(s,x)| < |S|\leq |T| < |S'|$.

\emph{Case 2: $s \in S'$.}
Now we know that $S'\subset S$. We show next that in Case~2 
$S$ is also a smallest SC with respect to $x$ and $t$ and $\SC(t,s) = \SC(t,x)$.

According to the correctness of the cut tree algorithm (Lemma~\ref{lem:cut_pairs}), 
$S$ is also induced by a minimum $x$-$t$-cut, which is $\lambda(x,t) = \lambda(s,t)$ and $S = \SC(x,t)$. 
Hence, it is further $\SC(t,s) = \SC(t,x)$, as $x\in S$,
and together with $S = \SC(x,t)$ (see above),
$S$ is a smallest SC with respect to $x$ and $u$, since
$|\SC(t,s)|\geq |S|$.
\qed
\end{proof}

\vspace{2ex}
\rephrase{Lemma}{\ref{lem:refineOutside}}{\lemRefineOutside}

\begin{proof}
We distinguish two cases, namely $s \in S'$ (Figure~\ref{fig:cutPairsNew_c},\ref{fig:cutPairsNew_d})
and $x\in S'$ (Figure~\ref{fig:cutPairsNew_e},\ref{fig:cutPairsNew_f}).
The first assertion, which is $U\cap S' = \emptyset$ or $U \subset S'$,
follows in both cases directly from Lemma~\ref{lem:intersecBehavior}(1),(2i), 
depending on whether $u\in S'$.
In the following we proof the further assertions.

\emph{Case 1:} $s\in S'$.
If $U\cap S' = \emptyset$, due to the correctness of the cut tree 
algorithm (Lemma~\ref{lem:cut_pairs}),
$U$ induces a minimum $u$-$x$-cut, and hence,
it is $\lambda(x,u) = \lambda(s,u)\leq \lambda(x,s)$ and $U = \SC(u,x)$.

(i) $x\in \SC(s,u)$: Suppose $s\notin \SC(x,u)$. 
Then it follows $\lambda(x,s) = \lambda(x,u)$
and $S' \subset \SC(s,u)$ according to Lemma~\ref{fig:intersecBehavior}(2i).
Hence, $S'$ would be a smaller SC with respect to $s$ and $u$ than 
$\SC(s,u)$, which is a contradiction.
Thus, it is $s\in \SC(x,u)$ and by Lemma~\ref{fig:intersecBehavior} (2ii),
it is $\SC(x,u) = \SC(s,u)$ and together with $U = \SC(u,x)$ (see above), 
$U$ is a smallest SC with respect to $u$ and $x$,
since $|\SC(s,u)| \geq |U|$.

(ii) $x\notin \SC(s,u)$: It follows $\lambda(x,s) = \lambda(s,u)$
and $\SC(s,u) = S'$ by Lemma~\ref{fig:intersecBehavior}(2ii).
If further $u\notin \SC(x,s)$, by Lemma~\ref{lem:intersecBehavior}(2ii),
we see that $\SC(x,s) = \SC(x,u)$.
It is further $|\SC(x,u)| \geq |S'| \geq |U|$, and thus, 
together with $U = \SC(u,x)$ (see above),
$U$ is a smallest SC with respect to $u$ and $x$.

\emph{Case 2:} $x\in S'$.
If $U\subset S'$, due to the correctness of the cut tree algorithm (Lemma~\ref{lem:cut_pairs}),
$U$ induces a minimum $u$-$x$-cut, and hence,
it is $\lambda(x,u) = \lambda(s,u)\leq \lambda(x,s)$ and $U = \SC(u,x)$.

We claim that $s \in \SC(x,u)$, which helps to prove (i) and (ii).
Suppose $s\notin \SC(x,u)$. Then it is $\lambda(x,s) = \lambda(x,u)$ and 
the cut inducing $\SC(x,u)$ can be bent along $S'$ deflected by $s$, 
such that $\SC(x,u) \subset S'$, and hence, induces
a smaller SC with respect to $s$ and $x$.

(i) $x\in\SC(s,u)$: Since also $s\in \SC(x,u)$,
it is $\SC(x,u) =\SC(s,u)$, according to Lemma~\ref{lem:intersecBehavior}(2ii).
Furthermore, together with $U = \SC(u,x)$ (see above), $U$ is a smallest SC with respect to $x$ and $u$, since $\SC|(s,u)|\geq |U|$.

(ii) $x \notin \SC(s,u)$: Hence, $\lambda(s,x) = \lambda(s,u)$, and thus, $\SC(s,u) = \SC(s,x)$.
With $s\in \SC(x,u)$ and $x\notin \SC(s,u)$, by Lemma~\ref{lem:intersecBehavior}(2i)
it follows $\SC(s,u)\subset \SC(x,u)$.
Hence, $|U| \leq \SC(s,u) < |\SC(x,u)|$ and together with $U = \SC(u,x)$ (see above) ,
$U$ is a smallest SC with respect to $u$ and $x$.
\qed
\end{proof}

\putbib
\end{bibunit}

\end{appendix}

\end{document}